\documentclass[10pt,reqno]{article}
\usepackage{amsmath,amssymb,amsthm}
\usepackage{esint}
\usepackage{graphicx,color}
\usepackage{algorithm}
\usepackage{subfigure}
\usepackage{algpseudocode}
\usepackage[sort&compress]{natbib}
\newcommand{\nm}{\medskip}

\DeclareMathAlphabet{\itbf}{OML}{cmm}{b}{it}

\def\by{{{\itbf y}}}
\def\bx{{{\itbf x}}}

\newcommand{\G}{\mathbf{G}}
\newcommand{\eps}{\epsilon}
\newcommand{\I}{{\mathbf{I}}}

\newcommand{\E}{\mathbf{E}}
\newcommand{\h}{\mathbf{H}}
\newcommand{\J}{{\mathbf{J}}}
\newcommand{\hE}{\widehat{\E}}
\newcommand{\hh}{\widehat{\h}}
\newcommand{\RR}{\mathbb{R}}

\newcommand{\bd}{\mathbf{d}}

\newcommand{\Itr}{\mathcal{J}_0}
\newcommand{\Ier}{\mathcal{J}_{a,\rho}}

\newcommand{\bp}{\mathbf{p}}

\newcommand{\R}{\mathbb{R}}

\newcommand{\K}{{\kappa}}
\newcommand{\dis}{\displaystyle}
\newcommand{\OL}{\mathcal{L}}

\newcommand{\W}{\mathcal{W}}

\newcommand{\ds}{\displaystyle}

\newcommand{\F}{\mathcal{F}}

\newcommand{\hG}{\widehat{\G}}

\newcommand{\OP}{{\mathcal{P}}}

\newcommand{\B}{\mathbf{B}}
\newcommand{\C}{\mathcal{C}}

\newtheorem{thm}{Theorem}[section]

\newtheorem{lem}[thm]{Lemma}

\newtheorem{rem}[thm]{Remark}

\numberwithin{equation}{section}

\newcommand{\pathfigures}{Figures/}
\graphicspath{{\pathfigures}}

\begin{document}
\title{Electromagnetic time reversal algorithms and source localization in lossy dielectric media}
\author{
Abdul Wahab\thanks{Department of Mathematics, COMSATS Institute of Information Technology, 47040, Wah Cantt., Pakistan (wahab@ciitwah.edu.pk, amerasheed@ciitwah.edu.pk, rabnawaz@ciitwah.edu.pk).}
\and  Amer Rasheed\footnotemark[1]
\and  Tasawar Hayat\thanks{Department of Mathematics, Quaid-i-Azam University, 45320, Islamabad, Pakistan (pensy\_t@yahoo.com).} \thanks{Nonlinear Analysis and Applied Mathematics (NAAM) Research Group, Faculty of Science, King Abdulaziz
University,21589,  Jeddah, Saudi Arabia.} \thanks{Address correspondence to T. Hayat, E-Mail: pensy\_t@yahoo.com, Tel.: +92 51 90642172, Fax: +92 51 2275341.} 
\and  Rab Nawaz\footnotemark[1]
}

\maketitle

\begin{abstract}
The problem of reconstructing the spatial support of an extended radiating electric current source density in a lossy dielectric medium from transient boundary measurements of the electric fields is studied.  A time reversal algorithm is proposed to localize a source density from loss-less wave-field measurements.  Further, in order to recover source densities in a lossy medium, we first build attenuation operators thereby relating loss-less waves with lossy ones. Then based on asymptotic expansions of attenuation operators with respect to attenuation parameter, we propose two time reversal strategies for localization. The losses in electromagnetic wave propagation are incorporated using the Debye's complex permittivity, which is well-adopted for low frequencies (radio and microwave) associated with polarization in dielectrics.
\end{abstract}

\noindent {\footnotesize {\bf PACS 2010.}  42.30.Wb; 02.30.Zz; 42.81.Dp;  41.20.Jb;   02.60.Cb}


\section{Introduction}

Time reversal algorithms have been an important tool  to solve inverse problems in science and engineering since their premier applications \cite{fink, borceaPapa02, josselin}. These algorithms exploit the time invariance and the reciprocity of non-attenuating waves which substantiate that a wave travels through a loss-less medium and converges at the location of its source (scatterer, reflector or emitter) on re-emission after reversing the time using transformation $t\to t_{\rm final}-t$.   The idea has been successively used in telecommunication \cite{Tele, MY-FT-06, MY-FT-08}, biomedical imaging  \cite{FinkTan09TrBioIm}, inverse scattering theory \cite{XuWang04, GG, Afzal, GWL}, non-destructive evaluation \cite{Carminati, tr2} and prospecting geophysics \cite{Wapenaar} for instance.

The robustness and simplicity of time-reversal techniques make them an impressive choice to resolve source localization problems. These problems have been the subject of numerous studies over the recent past due to a plethora of applications in diverse domains, especially in biomedical imaging,  non destructive testing and geophysics, see  for instance \cite{HAetal-11, HAetal-11b, noise, noise2, Localization, PAT, aggk, book, book1, book2, Vald, EEG-rev, Porter, note} and references therein. Several frameworks to recover spatial support of the stationary  acoustic, elastic and electromagnetic sources in time and frequency domain have been developed \cite{noise, PAT, Vald, source}, including time reversal algorithms \cite{fink, GG, GWL,  Carminati, HAetal-11, HAetal-11b, aggk}.   The inverse source problems are ill-posed having non-uniqueness issues generally due to the presence of non radiating sources \cite{Vald, Porter, Bleistein}. The stability and localization of radiating  electromagnetic sources with single frequency, multiple frequency and complete frequency bandwidth as well as transient data have been studied extensively, refer for instance to \cite{Porter, Bleistein, Albanese, bao, Bojarski, Turkel}.

An interesting problem in imaging is to model and compensate for the effects of wave attenuation on image quality. Most of the imaging techniques either emphasize a non-attenuating medium or do not adequately incorporate underlying phenomenon in reconstruction algorithms. As a consequence, one retrieves erroneous or less accurate wave synthetics which produce serious blurring in reconstructed images. This is further blended with intrinsic instability and uncertainty of the reconstruction. All together, these effects complicate attempts to track the key features of the image and result in unfortunate information loss, refer to \cite{PhD} for a detailed account of attenuation artifacts in imaging.

This investigation aims to establish time reversal algorithms for isotropic dielectric lossy media thereby retrieving extended radiating current sources using transient measurements of the  electric field over an imaging domain in attenuating environment. 

Unfortunately, the time-reversibility of waves is forsaken in lossy media thereby impeding classical time reversal algorithms to be applicable. Recently, Ammari et al. \cite{HAetal-11, HAetal-11b,noise, PAT, book1, book2} have extended the time reversal algorithms to attenuating acoustic and elastic media and to inverse source problems  using asymptotic expansions of so-called attenuation maps with respect to attenuation parameters.  Considering Stokes' \emph{thermo-viscous wave} model for attenuation two algorithms are implemented in acoustic and elastic media. First an \emph{adjoint wave time reversal} algorithm is established wherein the adjoint lossy wave is re-emitted into the medium. However, since the adjoint lossy wave is explosive in nature, indeed due to the exponentially growing component of the respective adjoint Green functions with frequency, a regularization using frequency truncation of the attenuation maps is discussed. Then, a \emph{pre-processing time reversal} algorithm is established based on a higher order asymptotic development with respect to attenuation parameter by virtue of stationary phase theorem. The asymptotic expansion is utilized to  filter the attenuated measurements and  subsequently the classical time reversal algorithm is invoked to back propagate the data. Since, the considered attenuation model lacks the \emph{causality} property, the results of these studies were extrapolated to more realistic causal power-law type attenuation models and were combined with variant time reversal strategies by \citet{otmarK} and  \citet{Kower}.

In this work, we concentrate on Debye's complex permittivity model for attenuation and leave the discussion on causality of the model and its generalizations to power-law models for future. Two situations are taken into account. We begin with a non-attenuating medium and afterwards focus on extended source recovery when the medium is lossy and obeys the Debye's law. We follow the approach by  \citet{HAetal-11b} for constructing imaging functions. In order to achieve time reversal in lossy media, the so-called attenuation maps are identified and their asymptotic developments with respect to Debye's attenuation parameter are established informally. The formal developments can still be achieved using theorem of stationary phases or of steepest descent as in \cite{PAT, otmarK, M2AS, hormander}, and will not be discussed.

The investigation is sorted in the following order. A few preliminary results and some key  identities are  collected in Section \ref{s:problem}. In Section \ref{s:Ideal}, an electromagnetic source is retrieved using transient measurements of the tangential component of electric field  in a loss-less medium. Section \ref{s:att} is dedicated to the construction of attenuation maps. Their asymptotic expansions are derived and  lossy time reversal algorithms are established. A few numerical illustrations are provided in Section \ref{s:num} to elucidate the pertinence of imaging functions proposed in this work. The principle contributions of the investigation are finally summarized in Section \ref{s:conc}.

\section{Preliminaries}\label{s:problem}

Let $\Omega\subset\RR^3$ be an open bounded domain with a Lipschitz boundary $\Gamma$. Consider the time-dependent homogeneous linear Maxwell equations
\begin{eqnarray}
\begin{cases}
\ds \nabla\times\E_0(\bx,t) +\mu_0\dfrac{\partial\h_0}{\partial t}(\bx,t)= 0, &(\bx,t)\in\RR^3\times\RR,
\\
\nm
\ds\nabla\times\h_0(\bx,t)-\eps_0\dfrac{\partial\E_0}{\partial t}(\bx,t)= \delta_0(t)\J(\bx), &(\bx,t)\in\RR^3\times\RR,
\\
\nm
\ds\E_0(\bx,t)=0=\dfrac{\partial\E_0}{\partial t}(\bx,t), & \bx\in\RR^3,t\ll 0,
\\
\nm
\ds\h_0(\bx,t)=\mathbf{0}=\dfrac{\partial\h_0}{\partial t}(\bx,t), & \bx\in\RR^3,t\ll 0,
\end{cases}
\label{EMt}
\end{eqnarray}
with electric permittivity $\eps_0>0$ and permeability $\mu_0>0$. $\E_0$ and  $\h_0$ are the electric and magnetic fields respectively and $\delta_{0}(t)$ is the Dirac mass at $t=0$. Here $\J(\bx)\in\RR^3$ is the radiating current source density. We assume that $\J(\bx)$ is sufficiently smooth and compactly supported in  $\Omega$, that is, $\rm{supp}\big\{\J\big\}\subset\subset\Omega$. 
 
By virtue of \eqref{EMt}, fields $\E_0$ and $\h_0$ are the solutions to,
\begin{equation}
\left\{
\begin{array}{ll}
\nabla\times\nabla\times\E_0(\bx,t)+\dis\frac{1}{c_0^2}\frac{\partial^2}{\partial t^2}\E_0(\bx,t)
=-\dis\mu_0\J(\bx)\frac{\partial\delta_0\left(t\right)}{\partial t}, & (\bx,t)\in\RR^3\times\RR, 
\\\nm
\E_0(\bx,t)=\mathbf{0}=\dis\frac{\partial\E_0}{\partial t}(\bx,t), & \bx\in\RR^3, t\ll 0,
\end{array} 
\right.\label{Et}
\end{equation}
and 
\begin{equation}
\left\{
\begin{array}{ll}
\nabla\times\nabla\times\h_0(\bx,t)+\dis\frac{1}{c_0^2}\frac{\partial^2}{\partial t^2}\h_0(\bx,t)
=\nabla\times \J(\bx)\delta_0\left(t\right), &(\bx,t)\in\RR^3\times\RR, 
\\\nm
\h_0(\bx,t)=\mathbf{0}=\dis\frac{\partial\h_0}{\partial t}(\bx,t), & \bx\in\RR^3, t\ll 0.
\end{array} 
\right.
\label{Ht}
\end{equation}

In the sequel, we refer to $\K_0:=\omega\sqrt{\eps_0\mu_0}= {\omega}/{c_0}$ as the wave number with $c_0:= {1}/{\sqrt{\eps_0\mu_0}}$ being the wave speed in dielectrics with frequency pulsation $\omega$. Furthermore, we denote by $\hat{v}(\omega)$ or $\F[v(\cdot)](\omega)$ the Fourier transform of a function $v(t)$ with the conventions
$$
\hat{v}(\omega)=\int_\RR v(t)e^{-i\omega t} dt
\quad\text{and}\quad
v(t)=\dfrac{1}{2\pi}\int_\RR \hat{v}(\omega) e^{i\omega t}d\omega.
$$

Let $\hE_0$ and $\hh_0$ be the time-harmonic electric and magnetic fields, that is, 
\begin{equation}
\begin{cases}
\nabla\times\hE_0+i\omega\mu_0\hh_0 =\mathbf{0}, &\bx\in\RR^3,
\\
\nabla\times\hh_0-i\omega\eps_0\hE_0= \J(\bx), & \bx\in\RR^3,
\end{cases}
\label{EMw}
\end{equation}
subject to the Silver-M\"uller radiation conditions
\begin{equation}
\ds\mathbf{0}=\lim_{|\bx|\to\infty} |\bx| 
\begin{cases}
\sqrt{\mu_0}\hh_0\times\hat{\bx}-\sqrt{\eps_0}\hE_0,
\end{cases}
\qquad\text{where}\quad\hat{\bx}:=\frac{\bx}{|\bx|}.
\label{radiation}
\end{equation}

Consequently, the time-harmonic fields $\hE_0$ and $\hh_0$ then satisfy the Helmholtz equations
\begin{eqnarray}
\begin{cases}
\nabla\times\nabla\times\hE_0-\K_0^2\hE_0= -i\omega\mu_0\J(\bx),& \bx\in\RR^3,
\\
\nabla\times\nabla\times\hh_0-\K_0^2\hh_0= \nabla\times\J(\bx),& \bx\in\RR^3,
\end{cases}
\end{eqnarray}
subject to the outgoing radiation conditions \eqref{radiation}.

\subsection{Electromagnetic fundamental solutions}\label{ss:em-fun}

In order to derive the time reversal algorithms for electromagnetic source imaging and to analyze their localization properties, we recall electromagnetic fundamental solutions and revisit some of their important features and properties.

Let $\hG_0^{\rm{ee}}(\bx,\omega)$ and $\hG_0^{\rm{me}}(\bx,\omega)$ be the outgoing electric-electric and  magnetic-electric time-harmonic Green functions for the Maxwell equations, that is
\begin{equation}
\left\{
\begin{array}{l}
\nabla \times \hG_0^{\rm{ee}}(\bx,\omega)-i\omega\mu_0\hG_0^{\rm{me}} (\bx,\omega)=0,
\\\nm
\nabla \times \hG_0^{\rm{me}}(\bx,\omega)+i\omega\eps_0 \hG_0^{\rm{ee}}(\bx,\omega)=\I \delta_{\mathbf 0}\left(\bx\right),
\end{array} 
\right.
\end{equation}
where $\I$ is $3\times 3$ identity matrix.
It is well-known, see for instance \cite{Hansen, nedelec}, that for all $\bx\neq\bf{0}$
\begin{equation}
\left\{
\begin{array}{l}
\hG_0^{\rm{ee}}(\bx,\omega)=i\omega\mu_0\left(\I+\dfrac{1}{\K_0^2}\nabla\nabla\right)\widehat{g_0}(\bx,\omega),
\\\nm
\hG_0^{\rm{me}}(\bx,\omega)=-\nabla\times\I \widehat{g_0}(\bx,\omega),
\end{array}
\right.
\end{equation}
where $\widehat{g_0}(\bx,\omega)$ is the fundamental solution to the Helmholtz operator $-(\Delta+\K_0^2)$ in $\RR^3$, subject to Sommerfeld's outgoing radiation conditions, given by 
\begin{equation}
\widehat{g_0}(\bx,\omega)=
%\begin{cases}
%\ds\dfrac{i}{4}H_0^{(1)}(\K_0|\bx|), & \qquad \bx\neq 0, %\bx\in\RR^2,
%\\
\ds\dfrac{1}{4\pi|\bx|}\exp\{i\K_0|\bx|\},  \qquad \bx\neq 0,  \bx\in\RR^3,
%\end{cases}
\end{equation}
where $H^{(1)}_0$ is the zeroth order Hankel function of first kind.

Let us define $\G_0^{\rm{ee}}(\bx,t)$ and $\G_0^{\rm{me}}(\bx,t)$ for all $\bx\in\RR^3$, and $\tau,t\in\RR$ by
\begin{eqnarray}
\G_0^{\rm{ee}}(\bx,t)&=&\F^{-1}\big[\hG_0^{\rm{ee}}(\bx,\omega)\big](t)
%\nonumber
%\\\nm
%&=&
=\dis\dfrac{1}{2\pi}\int_{\RR}\hG_0^{\rm{ee}}(\bx,\omega)e^{i\omega t}d\omega,
\label{Get}
\\\nm
\G_0^{\rm{me}}(\bx,t)&=&\F^{-1}\big[\hG_0^{\rm{me}}(\bx,\omega)\big](t)
%\nonumber
%\\\nm
%&=& 
=\dis\dfrac{1}{2\pi}\int_{\RR}\hG_0^{\rm{me}}(\bx,\omega)e^{i\omega t}d\omega.
\label{Gmt}
\end{eqnarray}

\begin{proof}[Spatial reciprocity]
It can be proved for isotropic dielectrics (see for instance \cite{Wapenaar}) that for all $\bx,\by\in\RR^3$, $\bx\neq \by$ and $t\in\RR$,
\begin{eqnarray}
\G_0^{\rm{ee}}(\bx-\by,t)=\G_0^{\rm{ee}}(\by-\bx,t)
\quad\text{and}\quad
\G_0^{\rm{me}}(\bx-\by,t)=\G_0^{\rm{me}}(\by-\bx,t).
\label{reciprocity}
\end{eqnarray}
\end{proof}

The following identities from \cite{IPSE, book2} are the key ingredients to elucidate the localization property of the imaging algorithms proposed in the subsequent sections.  

\begin{lem}[Electromagnetic Helmholtz-Kirchhoff identity]\label{EM-HKI}
Let $\B(0,R)$ be an open ball in $\RR^3$ with large radius $R\to\infty$ and boundary $\partial \B(0,R)$. Then, for all $\bx,\by\in\RR^3$, we have 
$$
\ds\lim_{R\to +\infty}\int_{\partial \B(0,R)}\hG_0^{\rm{ee}}(\bx-\xi,\omega)\overline{\hG_0^{\rm{ee}}}(\xi-\by,\omega)d\sigma(\xi)=\mu_0 c_0\Re e\bigg\{\hG_0^{\rm{ee}}(\bx-\by,\omega)\bigg\},
$$
where the superposed bar indicates a complex conjugate.
\end{lem}

\begin{lem}\label{lem2}
For all $\bx,\by\in\RR^3$, $\bx\neq\by$,
$$
\quad\ds\frac{\epsilon_0}{2\pi}\int_\RR\Re e\bigg\{\hG_0^{\rm{ee}}(\bx-\by,\omega)\bigg\}\,d\omega= \delta_\bx(\by)\I.
$$
\end{lem}

\section{Source reconstruction in ideal media}\label{s:Ideal}

Assume that we are able to collect the fields $\E_0$ and $\h_0$ for all $(\bx, t)\in\Gamma\times [0, T]$, for $T$  sufficiently large. 
If both  components on $\Gamma$ are time-reversed from the final time $T$ (using transformation $t\to T-t$) and re-emitted from $\Gamma$, two fields propagate inside $\Omega$ in time reverse chronology converging towards the source $\J(\bx)$ as $T-t\to 0$. The ultimate goal of this section is to use the convergence of the  back-propagating fields to identify support of the current density $\J(\bx)$. Precisely, the problem under consideration is the following:

\begin{proof}[Inverse source problem]
Given the measurements of the electric and magnetic fields, $\E_0$ and $\h_0$ satisfying \eqref{EMt}, over $\Gamma\times [0,T]$, for $T$ sufficiently large, find the support, $\rm{supp}\big\{\J\big\}$, of the current source density $\J$.
\end{proof}

\subsection{Time reversal of  electric field}

In the rest of this contribution, we concentrate only on the time reversal imaging functions related to the electric field  in order to identify $\rm{supp}\big\{\J\big\}$. The case of magnetic field can be dealt with analogously.

Let us introduce an adjoint wave $\E_0^s$, for a fixed $s\in[0,T]$, satisfying
\begin{equation}
\left\{
\begin{array}{ll}
\nabla\times\nabla\times{\E}_0^s(\bx,t)+\dis\frac{1}{c_0^2}\dfrac{\partial^2}{\partial t^2}{\E}_0^s(\bx,t)=-\dis \mu_0\bd_e(\bx,T-s)\frac{\partial\delta_s\left(t\right)}{\partial t}\delta_\Gamma, & (\bx,t)\in\RR^3\times\RR, 
\\
\nm
\dis {\E}_0^s(\bx,t)=\mathbf{0}\quad\text{and}\quad
\dis\frac{\partial{\E}_0^s}{\partial t}(\bx,t) =\mathbf{0}, & \bx\in\RR^3, \, t\ll s,
\end{array} 
\right.\label{Es}
\end{equation}
where $\delta_\Gamma$ is the Dirac mass at  $\Gamma$ and the data set 
$$
\W_{e}=\bigg\{\bd_e(\bx,t):=\E_0(\bx,t),\quad
\forall(\bx,t)\in\Gamma\times [0,T]\bigg\},
$$ 
contains the  electric field $\E_0$ where $T$ is large enough so that electric field and its time derivative almost vanish identically for all $t>T$.  

By definition, the adjoint wave can be represented for all $(x,t)\in\RR^3\times\RR$ as,
\begin{equation}
\E_0^s(\bx,t)=-\ds\int_{\Gamma}\G_0^{\rm{ee}}(\bx-\xi,t-s)\bd_e(\xi, T-s)d\sigma(\xi).
\label{EsGeg}
\end{equation}
Each one of these waves is associated with a datum collected at a particular time instance $t=s$ and therefore contributes to the reconstruction of the source on re-emission. In order to gather all the information about source distribution, we add up the adjoint waves $\E_0^s$ for all $s\in [0,T]$. Precisely, we define a time reversal imaging function by 
\begin{equation}
\label{Iem}
\Itr(\bx):= \ds\dfrac{\eps_0}{\mu_0c_0}\int_0^T{\E}_0^s(\bx,T) ds,\quad \bx\in\Omega,
\end{equation}
and claim that $\Itr(\bx)\simeq \J(\bx)$. Indeed, we have the following theorem.

\begin{thm}\label{thm:em}
Let $\Itr$ be the time reversal functional defined by \eqref{Iem}. For all $\bx\in\Omega$ far from the boundary $\Gamma$ (compared to the wavelength),
$$
\Itr(\bx)\simeq\J(\bx).
$$
\end{thm} 
\begin{proof}
Notice that since $\J$ is compactly supported in $\Omega$, and $T$ is sufficiently large so that the field is negligible outside $[0,T]$,
\begin{eqnarray*}
\widehat{\bd_e}(\xi,\omega)
&=& 
\F\left[{\bd_e}(\xi,\cdot)\right](\omega)
\\
&=&\F\left[\E_0(\xi,\cdot)\big|_{\Gamma\times[0,T]}\right](\omega),
\\
&\simeq &
\F\left[\E_0(\xi,\cdot)\big|_{\Gamma\times\RR}\right](\omega)
\\
&=&
-\F\left[\int_{\Omega}\G_0^{\rm{ee}}(\by-\xi,\cdot)\J(\by)d\by\big|_{\Gamma\times\RR}\right](\omega),
\\
&=&-\ds\int_{\RR^3} {\hG_0^{\rm{ee}}}(\xi-\by,\omega)\J(\by)d\by\big|_{\xi\in\Gamma}.
\end{eqnarray*}
Then by using \eqref{EsGeg} in \eqref{Iem} and by virtue of Perseval's identity, we have for all $\bx\in\Omega$ away from $\Gamma$,
\begin{eqnarray}
\Itr(\bx) &=&-\ds\dfrac{\eps_0}{\mu_0c_0}\iint_{[0,T]\times\Gamma}\G_0^{\rm{ee}}(\bx-\xi,t-s)\bd_e(\xi, T-s)d\sigma(\xi)ds,
\nonumber
\\
\nm
&=&-\dfrac{\eps_0}{2\pi\mu_0c_0}\ds\iint_{\RR\times\Gamma} \hG_0^{\rm{ee}}(\bx-\xi,\omega)\overline{\widehat{\bd}}_e(\xi, \omega)d\sigma(\xi)d\omega,
\nonumber
\\
\nm
&\simeq&\dfrac{\eps_0}{2\pi\mu_0c_0}\iiint_{\RR^3\times\RR\times\Gamma} \hG_0^{\rm{ee}}(\bx-\xi,\omega)\overline{\hG_0^{\rm{ee}}}(\xi-\by,\omega)\J(\by)d\sigma(\xi)d\omega d\by,
\nonumber
\\
\nm
&=&\dfrac{\eps_0}{2\pi}\iint_{\RR^3\times\RR} \left(\dfrac{1}{\mu_0 c_0}\int_\Gamma\hG_0^{\rm{ee}}(\bx-\xi,\omega)\overline{\hG_0^{\rm{ee}}}(\xi-\by,\omega)d\sigma(\xi)\right)\J(\by)d\omega d\by.
\end{eqnarray}

Now invoking Lemma \ref{EM-HKI}, we obtain 
\begin{eqnarray*}
\dfrac{1}{\mu_0 c_0}\int_\Gamma\hG_0^{\rm{ee}}(\bx-\xi,\omega)\overline{\hG_0^{\rm{ee}}}(\xi-\by,\omega)d\sigma(\xi)
\simeq  \Re  e\bigg\{\hG_0^{\rm{ee}}(\bx-\by,\omega)\bigg\},
\end{eqnarray*}
and therefore we have
\begin{eqnarray*}
\Itr(\bx)
&\simeq &
\ds\frac{\eps_0}{2\pi}\iint_{\RR^3\times\RR}  \Re  e\bigg\{\hG_0^{\rm{ee}}(\bx-\by,\omega)\bigg\} d\omega\J(\by) d\by,
\\
\nm
&=&
\ds\int_{\RR^3} \left(\frac{\eps_0}{2\pi}
\int_\RR \Re  e\bigg\{\hG_0^{\rm{ee}}(\bx-\by,\omega)\bigg\} d\omega\right)\J(\by) d\by,
\\
\nm
&=&
\ds\int_{\RR^3}\delta_\bx(\by)\J(\by)d\by,
\\\nm
&=&\J(\bx),
\end{eqnarray*}
where we have made use of the identity 
$$
\frac{\eps_0}{2\pi}\int_\RR\Re e\bigg\{\hG_0^{\rm{ee}}(\bx-\by,\omega)\bigg\} d\omega =\delta_\bx(\by)\I,
$$
from Lemma \ref{lem2}.
\end{proof}

\section{Source reconstruction in lossy media}\label{s:att}

In this section, we present a time reversal strategy for imaging in lossy media. We consider a Debye law to incorporate losses in wave propagation, which is suitable for low frequencies (radio to microwave) associated with polarization in dielectrics \cite{MY-FT-06, MY-FT-08, Koledintseva}. We will only consider the electric case. 

Let $\hE_a(\bx,\omega)$ be the electric field in a lossy dielectric medium, that is, the solution to 
\begin{equation}
\nabla\times\nabla\times\hE_{a}(\bx,\omega)-\left(\K^\sigma_{a}(\omega)\right)^2\hE_{a}(\bx,\omega)=-i\omega\mu_0\J(\bx), \quad (\bx,\omega)\in\RR^3\times\RR,
\label{E-att-eq-f}
\end{equation}
subject to Sommerfeld radiation condition, 
$$
\lim_{|\bx|\to\infty} |\bx| \bigg|\nu \times \nabla  \times \hE_a(\bx,\omega) - i\K_a^\sigma(\omega)\hE_a(\bx,\omega)\bigg|= \mathbf{0},
$$
where $\nu$ is the outward unit normal at $\Gamma$ and $\K_{a}^\sigma(\omega)=\omega\sqrt{\mu_0\eps_{a}^\sigma}$ is the wave-number defined in terms of the Debye's complex permittivity:
\begin{equation}
\eps_{a}^\sigma= \eps_\infty +\dfrac{\eps_s-\eps_\infty}{1+ia\omega}-\dfrac{i\sigma}{\omega\epsilon_0}.
\end{equation}
Here $\eps_s$, $\eps_\infty$, $\sigma$ and $a$ are respectively, the static and infinite-frequency permittivities, electric conductivity and Debye's loss constant. We precise that $\eps_\infty\leq\eps_s$.
Furthermore, we fix 
$$
\E_a(\bx,t):=\F^{-1}[\hE_a(\bx,\cdot)](t)=
\dis\int_\RR \hE_a(\bx,\omega)e^{i\omega t}d\omega.
$$

Let $\widehat{\G}_{a}^{\rm{ee}}(\bx,\omega)$ be the attenuating electric-electric Green function, \emph{i.e.} the outgoing  fundamental solution to the lossy Helmholtz equation
\begin{equation}
\nabla\times\nabla\times\widehat{\G}_{a}^{\rm{ee}}(\bx,\omega)-\left(\K^\sigma_{a}(\omega)\right)^2\widehat{\G}_{a}^{\rm{ee}}(\bx,\omega)=i\omega\mu_0\delta_{\mathbf 0}(\bx)\I.
\label{Gee-att-eq-f}
\end{equation}

In the sequel,  we use the following notation
$$
{\G}_a^{\rm{ee}}(\bx,t):=\F^{-1}\left[\widehat{\G}_a^{\rm{ee}}(\bx,\cdot)\right](t)
=
\dis\int_\RR \widehat{\G}_a^{\rm{ee}}(\bx,\omega)e^{i\omega t}d\omega,
$$
and define $\epsilon_{-a}^{-\sigma}$ and $\K_{-a}^{-\sigma}$ and the Green functions  $\hG_{-a}^{\rm{ee}}$ and $\G_{-a}^{\rm{ee}}$ analogously.

\subsection{Attenuation operators}

Recall that 
\begin{equation}
\nabla \times \nabla \times \hG_0^{\rm{ee}}(\bx,\omega) - \frac{\omega^2}{c_0^2}\hG_0^{\rm{ee}}(\bx,\omega)= i\omega\mu_0\delta_{\mathbf{0}}(\bx). 
\label{Gee-eq-f}
\end{equation}
Therefore, by replacing real frequency $\omega$ with $c_0\K_a^\sigma(\omega)$ by invoking Theorem of Titchmarsh \cite{titchmarsh} since $\Im m\{c_0\K_a^\sigma(\omega)\}>0$, using an argument of the unique outgoing fundamental solution and comparing with Equation \eqref{Gee-att-eq-f}, we deduce that
\begin{equation*}
{\hG}_0^{\rm{ee}}\big(\bx,c_0\K_a^\sigma(\omega)\big)= \dfrac{c_0\K_a^\sigma(\omega)}{\omega}\widehat{\G}_a^{\rm{ee}}(\bx,\omega),\quad\forall \bx\in\RR^3,~\bx\neq\bf{0},
\end{equation*}
or equivalently
\begin{equation*}
{\G}_a^{\rm{ee}}(\bx,t)=\OL_a\left[{\G}_0^{\rm{ee}}(\bx,\cdot)\right](t).
\end{equation*}
Here we define the operator $\OL_a$ (hereafter called attenuation operator) %for a function $\phi\in\mathbb{S}'[0,\infty)$ 
as follows:
\begin{eqnarray}
\label{eq:OL}
\OL_a: && \mathbb{S}'([0,\infty))\longrightarrow \mathbb{S}'(\RR)
\nonumber
\\\nm
&& \phi(t)\longmapsto \ds\dfrac{1}{2\pi}\int_\R\dfrac{\omega}{c_0\K_a^\sigma(\omega)}
\left(\int_{\R^+} \phi(\tau)\exp\big\{-ic_0\K_a^\sigma(\omega) \tau\big\}d\tau\right) e^{i\omega t}d\omega,
\end{eqnarray}
where  $\mathbb{S}$ is the Schwartz space of rapidly decreasing functions and $\mathbb{S}'$ is the space of tampered distributions.

Let us also introduce operator ${\OL}_{-a,\rho}$ related to $\K_{-a}^{-\sigma}(\omega)$ and its adjoint operator ${\OL}^*_{-a,\rho}$  for all $\rho>0$ by:
\begin{eqnarray}
\label{eq:OLa}
{\OL}_{-a,\rho}: && \mathbb{S}'([0,\infty))\longrightarrow \mathbb{S}'(\RR)
\nonumber
\\\nm
&& \phi(t)\longmapsto \dfrac{1}{2\pi}\int_{\R^+}\phi(\tau)
\int_{\omega\leq \rho}\dfrac{\omega e^{i\omega t}}{c_0\K_{-a}^{-\sigma}(\omega)}\,\exp\big\{-ic_0\K_{-a}^{-\sigma}(\omega) \tau\big\}d\omega\,d\tau,
\end{eqnarray}
and 
\begin{eqnarray}
\label{eq:OLt}
{\OL}^*_{-a,\rho}: && \mathbb{S}'([0,\infty))\longrightarrow \mathbb{S}'(\RR)
\\
\nm
\nonumber
&& \phi(t)\longmapsto \dfrac{1}{2\pi}\int_{\omega\leq \rho} \dfrac{\omega}{c_0\K_{-a}^{-\sigma}(\omega)}
\left(\int_{\R^+} \phi(\tau)e^{i\omega \tau}d\tau\right)\exp\big\{-ic_0\K_{-a}^{-\sigma}(\omega)t\big\}d\omega.
\end{eqnarray}

We extend operators $\OL_a$, ${\OL}_{-a,\rho}$ and  ${\OL}^*_{-a,\rho}$ to $\G_0^{\rm{ee}}$, that is, for all constant vectors $\bp\in\R^3$, we define
$$
{\OL_a}[\G_0^{\rm{ee}}]\bp={\OL_a}[\G_0^{\rm{ee}}\bp],
$$

$$
{\OL}_{-a,\rho}[\G_0^{\rm{ee}}]\bp={\OL}_{-a,\rho}[\G_0^{\rm{ee}}\bp],
$$
and
$$
{\OL}^*_{-a,\rho}[\G_0^{\rm{ee}}]\bp={\OL}^*_{-a,\rho}[\G_0^{\rm{ee}}\bp].
$$

\subsection{Asymptotic analysis of attenuation operators}\label{ss:em-att}

We assume that the attenuation parameter $a$ is sufficiently small compared to the wave-length (denoted by $\lambda$), that is, 
$$
a\ll{c_0}{\omega}^{-1}=:\lambda.
$$
For brevity, we consider the  case of a non-conductive medium, that is, $\sigma=0$. Henceforth, we drop the superscript from $\K_a^\sigma$ and $\K_{-a}^{-\sigma}$ for  simplicity. Then 
\begin{eqnarray*}
c_0\K_a(\omega)&=& \dfrac{\omega}{\sqrt\eps_0}
\dis\sqrt{\eps_\infty+\dfrac{(\eps_s-\eps_\infty)}{1+i\omega a}},
\\
\nm
&\simeq&
\dfrac{\omega}{\sqrt\eps_0}
\dis\sqrt{\eps_\infty+(\eps_s-\eps_\infty)\left[1-(i\omega a)
%+ (i\omega a)^2
\right]+o(\omega a)},
\\
\nm
&\simeq&
\gamma\omega
\dis\sqrt{1-i\beta\omega a
%+ \omega^2 a^2
+o(\omega a)},
\\
\nm
&\simeq&
\gamma\omega
\left(1-i\frac{\beta}{2}\omega a 
\right)+o(\omega a),
\end{eqnarray*} 
and 
$$
c_0\K_{-a}(\omega)\simeq \gamma\omega
\left(1+i\frac{\beta}{2}\omega a\right)
%+\frac{\beta(4-\beta)}{8}\omega^2 a^2
+o(\omega a),
$$
where  $\beta=\big(1-(\eps_\infty/\eps_s)\big)$ and $\gamma=\sqrt{{\eps_s}/{\eps_0}}$. 
Similarly, we have 
\begin{eqnarray*}
\dfrac{\omega}{c_0\K_a(\omega)}
=
\dfrac{1}{\gamma\sqrt{1-i\beta\omega a
%+ \omega^2 a^2
+o(\omega a)}}
\simeq
\dfrac{1}{\gamma}\left(1+i\dfrac{\beta}{2}\omega a\right)+o(\omega a),
\end{eqnarray*}
and 
\begin{eqnarray*}
\dfrac{\omega}{c_0\K_{-a}(\omega)}
=
\dfrac{1}{\gamma\sqrt{1+i\beta\omega a
%+ \omega^2 a^2
+o(\omega a)}}
\simeq
\dfrac{1}{\gamma}\left(1-i\dfrac{\beta}{2}\omega a\right)+o(\omega a).
\end{eqnarray*}
Then the following result holds.

\begin{lem}\label{lem-A-asymp}

Let $\phi\in \C_{0}^\infty\left([0,\infty)\right)$, where $\C_{0}^\infty$ is the space of $\C^\infty-$functions with compact support in $[0,\infty)$. Then, 

\begin{enumerate}
\item Up to leading order of attenuation parameter $a$
$$
\OL_a\left[\phi(\cdot)\right](t)\simeq
\dfrac{1}{\gamma^2}\phi\left(\dfrac{t}{\gamma}\right)
+\dfrac{\beta a}{2\gamma^3}\big[\phi'+(t\phi)''\big]\left(\dfrac{t}{\gamma}\right)\quad\text{as}\quad a\to 0,
$$ 

\item for all $\rho>0$, up to leading order of attenuation parameter $a$
$$
{\OL}_{-a,\rho}^*\left[\phi(\cdot)\right](t)\simeq
\dfrac{1}{\gamma}\OP_{\rho}[\phi(\cdot)]\left({\gamma{t}}\right)
-\dfrac{\beta a}{2\gamma^2}\OP_{\rho}\big[\phi'+(t\phi)''\big]\left({\gamma}{t}\right)
\quad\text{as}\quad a\to 0,
$$ 

\item for all $\rho>0$, up to leading order of attenuation parameter $a$
$$
\OL_{-a,\rho}^*\left[\OL_a\left[\phi(\cdot)\right]\right](t)\simeq
\dfrac{1}{\gamma^3}\OP_{\rho}[\phi(\cdot)](t) \quad\text{as}\quad a\to 0,
$$ 
\end{enumerate}
where $\OP_{\rho}$ is defined by
\begin{eqnarray}
\OP_{\rho}: &&\mathbb{S}'(\RR)\longrightarrow \mathbb{S}'(\RR)
\nonumber
\\
&&\phi(t)\longmapsto\frac{1}{2\pi} \int_{|\omega|\leq \rho} e^{- i \omega t } \F[\phi](\omega) d\omega.
\label{eq:P}
\end{eqnarray}
\end{lem}

\begin{proof}
We prove only Statements $1$ and $2$. Statement $3$ is an immediate consequence of the first two.

\begin{enumerate}
\item  As $a\to 0$ the attenuation operator $\OL_a$ can be approximated by:
\begin{eqnarray*}
\OL_a[\phi](t)
&\simeq&\dfrac{1}{2\pi\gamma}\int_\R
\left(1+i\dfrac{\beta}{2}\omega a\right)
\left\{\int_{\R^+}e^{-\gamma\frac{\beta}{2}a\omega^2 \tau}\phi(\tau)e^{-i\gamma\omega \tau }\,d\tau\right\} e^{i\omega t}d\omega+o(a),
\\
\nm
&\simeq&
\dfrac{1}{2\pi\gamma}\int_\R
\left(1+i\dfrac{\beta}{2}\omega a\right)
\left\{\int_{\R^+}\left(1-\gamma\dfrac{\beta}{2}a\omega^2 \tau\right)\phi(\tau)e^{-i\gamma\omega \tau}\,d\tau\right\} e^{i\omega t}\,d\omega +o(a),
\\
\nm
&\simeq&
\dfrac{1}{2\pi\gamma} 
\iint_{\R\times\R^+}\phi(\tau)e^{-i\gamma\omega \tau}
e^{i\omega t}\,d\tau\,d\omega
\\
\nm
&&\quad+
\dfrac{\beta a}{2}\dfrac{1}{2\pi\gamma}
\int_\R i\omega\int_{\R^+}\phi(\tau)e^{-i\gamma\omega \tau}
e^{i\omega t}\,d\tau\,d\omega
\\
\nm
&&
\quad
+
\dfrac{\beta a}{2}\dfrac{1}{2\pi}
\int_\R (i\omega)^2\int_{\R^+}\left[\tau\phi(\tau)\right]e^{-i\gamma\omega \tau}
e^{i\omega t}\,d\tau\,d\omega +o(a),
\\
\nm
&\simeq&
\dfrac{1}{\gamma^2}\phi\left(\dfrac{t}{\gamma}\right)
+\dfrac{\beta a}{2\gamma^3}\big[\phi'+(t\phi)''\big]\left(\dfrac{t}{\gamma}\right) + o(a),
\end{eqnarray*}
which is the required result.

\item 
Let the support of $\phi$ be contained in $\left[0,t_{\rm{max}}\right]\subsetneq [0,\infty)$. As $a\to 0$ the operator ${\OL}_{-a,\rho}^*$ can be approximated by:
\begin{eqnarray*}
{\OL}_{-a,\rho}^*\left[\phi(\cdot)\right](t)
&\simeq&\dfrac{1}{2\pi\gamma}
\int_{|\omega|\leq \rho}\dfrac{\omega}{c_0\K_{-a}(\omega)}
\exp\big\{-ic_0\K_{-a}(\omega)t\big\}
\\
\nm
&&
\qquad\qquad\qquad
\times
\left\{\int^{t_{\rm{max}}}_0 \phi(\tau)e^{i\omega \tau}\,d\tau\right\}d\omega
+o(a),
\\
\nm 
&\simeq&
\dfrac{1}{2\pi\gamma}\int_{|\omega|\leq \rho}\int^{t_{\rm{max}}}_0 \left(1-i\dfrac{\beta}{2}\omega a\right)
\\
\nm
&&
\qquad\qquad\times
\exp\left\{-i\gamma\omega\left(1+i\frac{\beta}{2}\omega a\right)t\right\}
\phi(\tau)e^{i\omega \tau}\,d\tau\,d\omega
+o(a),
\end{eqnarray*}
where we have made use of the approximation of lossy wavenumber $\K_{-a}$ for $a\ll c_0\omega^{-1}$. 

On further simplifications, we arrive at 
\begin{eqnarray*}
{\OL}_{-a,\rho}^*\left[\phi(\cdot)\right](t)
&\simeq&
\dfrac{1}{2\pi\gamma}\iint_{[-\rho,\rho]\times[0,t_{\rm{max}}]}  \left(1-i\dfrac{\beta}{2}\omega a\right)
e^{\gamma\frac{\beta}{2}\omega^2 a t}
\\\nm
&& \qquad\qquad\qquad\times
\phi(\tau)e^{-i\gamma\omega t+i\omega \tau}\,d\tau\,d\omega 
+o(a),
\\
\nm
&\simeq&
\dfrac{1}{2\pi\gamma}\iint_{[-\rho,\rho]\times[0,t_{\rm{max}}]}  \left(1-i\dfrac{\beta}{2}\omega a\right)
\left(1+\gamma\dfrac{\beta}{2}a\omega^2 t\right)
\\
&&\qquad\qquad\qquad\times
\phi(\tau)e^{-i\gamma\omega t}\,e^{i\omega\tau}\,d\tau\,d\omega +O(a),
\\
\nm
&\simeq&
\dfrac{1}{2\pi\gamma}\iint_{[-\rho,\rho]\times\R^+} 
\phi(\tau)e^{-i\gamma\omega t}e^{i\omega \tau}\,d\tau\,d\omega
\\
\nm
&&
- 
\dfrac{\beta a}{2}\dfrac{1}{2\pi\gamma}
\int_{|\omega|\leq \rho} i\omega\int_{\R^+}
\phi(\tau)e^{-i\gamma\omega t}e^{i\omega \tau}\,d\tau\,d\omega
\\
\nm
&&
-
\dfrac{\beta a}{2}\dfrac{1}{2\pi}
\int_{|\omega|\leq \rho} (i\omega)^2\int_{\R^+}\left[\tau\phi(\tau)\right]e^{-i\gamma\omega t}
e^{i\omega \tau}\,d\tau\,d\omega +o(a).
\end{eqnarray*}

Finally on introducing the operator $\OP_\rho$, we arrive at
\begin{eqnarray*}
{\OL}_{-a,\rho}^*\left[\phi(\cdot)\right](t)
&\simeq&
\dfrac{1}{\gamma}\OP_{\rho}[\phi(\cdot)]\left({\gamma{t}}\right)
-\dfrac{\beta a}{2\gamma^2}\OP_{\rho}\big[\phi'+(t\phi)''\big]\left({\gamma}{t}\right) + o(a).
\end{eqnarray*}
\end{enumerate}
\end{proof}

\begin{rem}
We precise that the Lemma \ref{lem-A-asymp}, can be proved formally using the argument of stationary phase theorem or steepest decent theorem as in \cite{PAT, otmarK, hormander}.
\end{rem}

\subsection{Time reversal of the electric field in lossy media}\label{ss:em-tr-att}

Suppose we are able to collect the attenuated electric field, $\E_{a}$, for all $t\in[0,T]$ over $\Gamma$, that is we are in possession  of the data set 
$$
\W_{e,a}:=\bigg\{\bd_{e,a}(\bx,t):=\E_{a}(\bx,t) : (\bx,t)\in\Gamma\times [0,T]\bigg\}.
$$
If we simply time-reverse and re-emit the measured data $\bd_{e,a}$ in attenuating medium, the electric field is attenuated twice, that is, both in direct and back-propagation. Therefore, resolution of the reconstruction, when localizing the spatial support of the sources, is forsaken. In this section, we present two time reversal strategies for localizing $\rm{supp}\big\{\J\big\}$ in a lossy medium. 

\subsubsection{Adjoint operator approach}

In order to compensate for losses, we back propagate the measured data using the adjoint wave operator. Unfortunately the adjoint wave problem is  severely ill-posed, somewhat similar phenomenon was observed in acoustic and elastic media. Therefore, high frequencies must be suppressed as in the acoustic and elastic cases; refer to \cite[Remark 2.3.6]{PhD}. More precisely, let 
$$
\E_{a}^s(\bx,t):=\F^{-1}\left[\hE_a^s(\bx,\cdot)\right](t),
$$ 
be the adjoint (time reversed) field corresponding to datum $\bd_{e,a}(\bx,t)$ recorded at time $t=s$ propagating inside the medium, where $\hE_a^s(\bx,\omega)$ is the solution to adjoint lossy Helmholtz equation
$$
\left(\nabla\times\nabla\times \hE_a^s-\K^2_{- a}\hE_a^s\right)(\bx,\omega)=-i\omega\mu_0\overline{\bd}_{e,a}(x,\omega)\delta_{\Gamma}(\bx),\quad (\bx,\omega)\in\RR^3\times\RR,
$$ 
and let  
$$
\E_{a,\rho}^s(\bx,t):= \OP_\rho\left[\E_{a}^s(\bx,\cdot)\right](t),
$$ 
where $\rho$ is the cutoff frequency. Here $\rho$ is chosen in such a way that $\E^s_{a,\rho}$ does not explode whereas the resolution of the time reversal algorithm remains reasonably intact, refer to  \cite[Remark 2.3.6]{PhD} for further details. The aim of this section is to justify that $\Ier(\bx)$ is an approximation of $\J(\bx)$ up to leading order of Debye's loss parameter $a$, when $\rho\to +\infty$, where
\begin{equation}
\label{em:Itra}
\ds\Ier(\bx):= \dfrac{\eps_0\gamma^3}{\mu_0c_0}\int_0^T \E_{a,\rho}^s(\bx,T)ds, \quad\forall \bx\in\Omega.
\end{equation}

We conclude this subsection with the following key result. It simply states that the adjoint operator approach provides a localization of the $\rm{supp}\big\{\J\big\}$  with a correction to the attenuation effects up to leading order of the damping parameter $a$.

\begin{thm}\label{thm:em-att}
For all $\bx\in\Omega$ sufficiently far from $\partial\Omega$, compared to wavelength, the truncated time-reversal imaging functional $\Ier$ satisfies, 
$$
\Ier(\bx)=\mathcal{J}_{0,\rho}(\bx)+O(a),
$$
where 
$$
\mathcal{J}_{0,\rho}(\bx) := -\ds\dfrac{\eps_0}{\mu_0c_0}\iint_{[0,T]\times \Gamma}\G^{\rm{ee}}_0(\bx-\xi,\tau)\OP_\rho\left[\bd_e(\xi,\cdot)\right](\tau)d\sigma(\xi)d\tau.
$$
Moreover, 
$$
\mathcal{J}_{0,\rho}(\bx)\to \J(\bx)\quad\text{as}\quad \rho\to+\infty.
$$
\end{thm}
\begin{proof} 
Notice that 
\begin{equation}
{\G}_{-a,\rho}^{\rm{ee}}(\bx,t):=\OP_\rho[{\G}_{-a}^{\rm{ee}}(\bx,\cdot)](t)={\OL}_{-a,\rho}\left[{\G}_0^{\rm{ee}}(\bx,\cdot)\right](t).\label{key}
\end{equation}
By virtue of \eqref{key},  $\Ier(\bx)$ can  be rewritten in the form
\begin{eqnarray*}
\Ier(\bx) = -\dfrac{\eps_0\gamma^3}{\mu_0c_0}\iint_{\Gamma\times[0,T]} \ds\G^{\rm{ee}}_0(\bx-\xi,s) \OL_{-a,\rho}^\ast\left[\bd_{e,a}(\xi,\cdot)\right](s) ds d\sigma(\xi).
\end{eqnarray*}

Remark as well that $\bd_{e,a}(x,t) = \OL_a\left[\bd_{e}(\bx,\cdot)\right](t)$, where $\bd_{e}(\bx,t)$ represents the ideal data, so that
\begin{eqnarray*}
 \Ier(\bx) &=& -\dfrac{\eps_0\gamma^3}{\mu_0c_0}\iint_{\Gamma\times[0,T]} \ds \G^{\rm{ee}}_0(\bx-\xi,s) \OL_{-a,\rho}^\ast\bigg[\OL_{a}\left[\bd_{e}(\xi,\cdot)\right]\bigg](s)ds d\sigma(\xi), 
 \\\nm
 &=&- \ds \dfrac{\eps_0}{\mu_0c_0}\iint_{\Gamma\times[0,T]}\ds\G^{\rm{ee}}_0 (\bx-\xi,s)
\OP_{\rho}\left[\bd_{e}(\xi,\cdot)\right](s)ds d\sigma(\xi)+o(a),
\\\nm
&=& \mathcal{J}_{0,\rho}(\bx)+o(a),
\end{eqnarray*}
by using Lemma \ref{lem-A-asymp}. Finally, from Theorem, \ref{thm:em}
\begin{eqnarray*}
-\dfrac{\eps_0}{\mu_0c_0}\iint_{\Gamma\times[0,T]} \ds\G^{\rm{ee}}_0(\bx-\by,s)  \OP_{\rho}\left[\bd_{e}(\xi,\cdot)\right](s) d\sigma(\xi) ds
&\underrightarrow{\rho\to\infty}& \Itr(\bx)\simeq \J(\bx),
\end{eqnarray*}
when $\bx$ is far away from the boundary $\Gamma$. The conclusion follows immediately.
\end{proof}

\subsubsection{Pre-processing approach}

According to Lemma \ref{lem-A-asymp}, for weakly attenuating media up to leading order
\begin{equation}
\OL_a\left[\phi(\cdot)\right](t)\simeq
\dfrac{1}{\gamma^2}\phi\left(\dfrac{t}{\gamma}\right)
+\dfrac{\beta a}{2\gamma^3}\big[\phi'+(t\phi)''\big]\left(\dfrac{t}{\gamma}\right).
\end{equation}
Therefore, its first order approximate inverse, $\OL^{-1}_{a,1}$, can be given by 
\begin{equation}
\OL^{-1}_{a,1}\left[\phi(\cdot)\right](t)
= \dfrac{1}{\gamma}\phi(\gamma t) -\frac{\beta a}{2\gamma^2}\left(\phi'+(t\phi)''\right)(\gamma t).
\end{equation}
In the similar fashion, using higher order asymptotic expansions, $k^{th}$ order approximate inverse $\OL^{-1}_{a,k}$ can be constructed.
Therefore, a pre-processing approach to time reversal can be described in two steps:
\begin{itemize}
\item[1. ] Filter the measured data in order to compensate for the attenuation effects using $\OL^{-1}_{a,k}$.
\item[2. ] Use classical time reversal (in ideal medium) with filtered data as input.

\end{itemize}
\begin{algorithm}
\caption{Pre-processing Time Reversal Algorithm: $k^{th}$ Order}
\label{EM:Itr}
\begin{algorithmic}[1]
\Require $\W_{e,a} =\bigg\{ \bd_{e,a}(x,t):=\E_a(x,t): \forall (x,t)\in\Gamma\times[0,T]\bigg\}$,\quad $0<a\ll c_0\omega^{-1}$ and $k\geq 1$.
\Procedure{Filter}{Pre-process $\bd_{e,a}(x,t)$.}
  
    \Return $\bd_{e}(x,t):=\OL_{a,k}^{-1}\big[\bd_{e,a}(x,\cdot)\big](t)$.
    
\EndProcedure

\Procedure{Time-Reversal}{Evaluate $\Itr(x)$.}
	
	\For{each $s\in [0,T]$}
		
		\State Construct $\E_0^s(x,T)$ for $x\in\Omega$ using  $\bd_{e}(x,t)$.
	\EndFor
	
	\State Evaluate $\Itr(x):=\ds\int_0^T \E_0^s(x,T)ds$.
	
	\Return $\Itr(x)$.
	
\EndProcedure
\end{algorithmic}

\Return $\mathcal{J}_{a,k}=\Itr(x)+o\left(a^k\right)$.
\end{algorithm}

\begin{rem}\label{rem-em-post} 
Pre-processing approach has two principle advantages over adjoint approach: 
\begin{enumerate}
\item It allows for higher order corrections to attenuation artifacts. Indeed using higher order approximations of the operator $\OL_{a}$, using stationary phase theorem \cite{hormander}, one can iteratively construct higher order pseudo-inverse $\OL^{-1}_{a,k}$. 
Consequently  filtered data using $\OL^{-1}_{a,k}$  yield a $k^{th}$ order correction for the attenuation artifacts. In this context, we refer to \cite{PAT, otmarK, HAetal-11b} 
\item It is much more stable numerically than the adjoint operator approach, as it has been observed for the case of acoustic imaging; refer to \cite{HAetal-11b, HAetal-11, PhD, book2}.
\end{enumerate}
\end{rem}

\section{Numerical illustrations}\label{s:num}

The aim here is to numerically illustrate the appositeness of the algorithms proposed in the previous sections. For brevity, we   consider the case of axis invariance (along $z-axis$) and restrict ourselves to a transverse electric case. In the sequel, we assume $\bx=(x_1,x_2,0)\in\Omega\subset X$ where  $X = \big[-l/2,l/2\big]\times \big[-l/2,l/2\big]\times\{0\}$ with periodic boundary conditions. For simplicity, we take $\eps_0=1=\eps_s$ and $\mu_0=1$, and consequently $c_0=1=\gamma$. Furthermore, we choose $\eps_\infty=0.5$ and therefore $\beta=0.5$. In order to numerically resolve the initial value problem \eqref{Et}, we use a \emph{Fourier spectral splitting} approach \cite{spectral} together with a  \emph{perfectly matched layer (PML)} technique \cite{pml} to simulate a free outgoing interface blended with the \emph{Strange's splitting} method \cite{strang}.

\begin{proof}[\textbf{Example 1}]

We choose $\Omega'$ to be a unit disk centered at origin such that $\Omega=\Omega'\times\{0\}$ and  place $1024$ equi-distributed sensors on its boundary.  We computed  the  electric fields over $(\bx,t) \in X\times\big[0,T\big]$ with $l = 4$ and $T=2$  and the space and time discretization steps are taken respectively $\tau = 2^{-13}T$ and $h= 2^{-9}l$.
In Figure \ref{f1}, a current source reconstruction using time reversal function $\Itr$ is compared to the initial current source density in a loss-less dielectric medium. The reconstructions clearly substantiate the accuracy and a high resolution of the time reversal algorithm $\Itr$. 

\end{proof}

\begin{figure}[!tbh]
\begin{center}
\includegraphics[width = 0.45\textwidth]{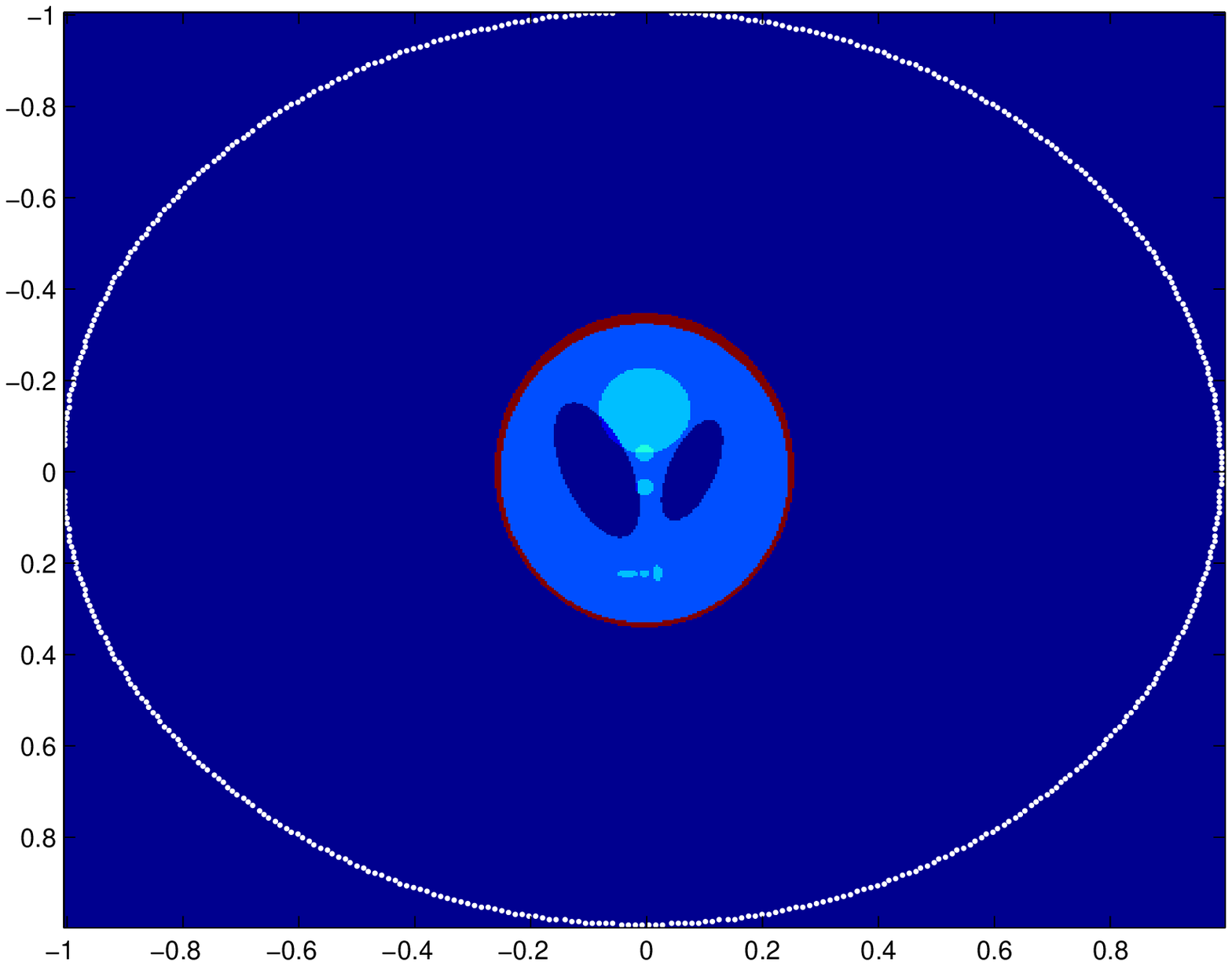}
\qquad
\includegraphics[width = 0.45\textwidth]{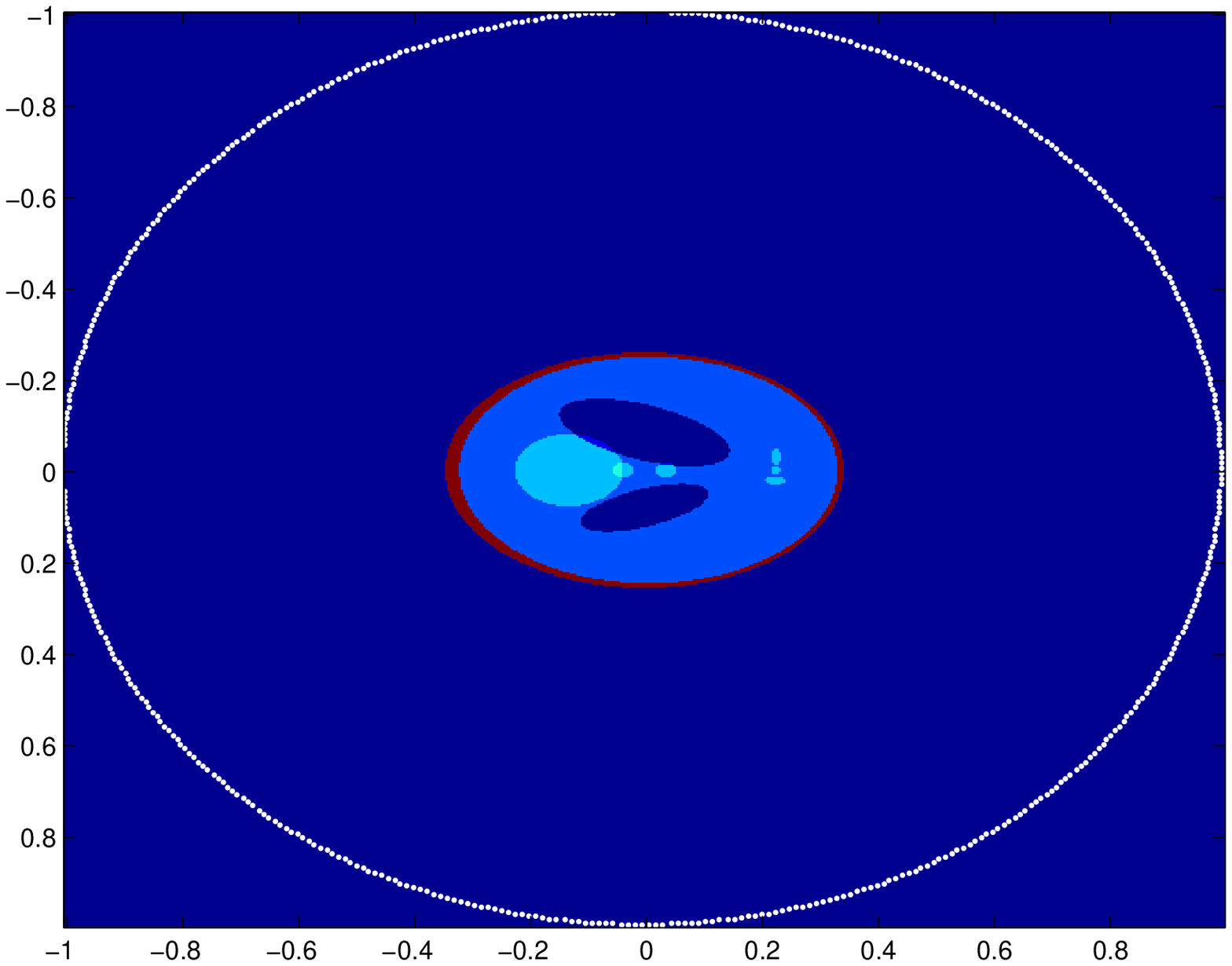}
\\
\includegraphics[width = 0.45\textwidth]{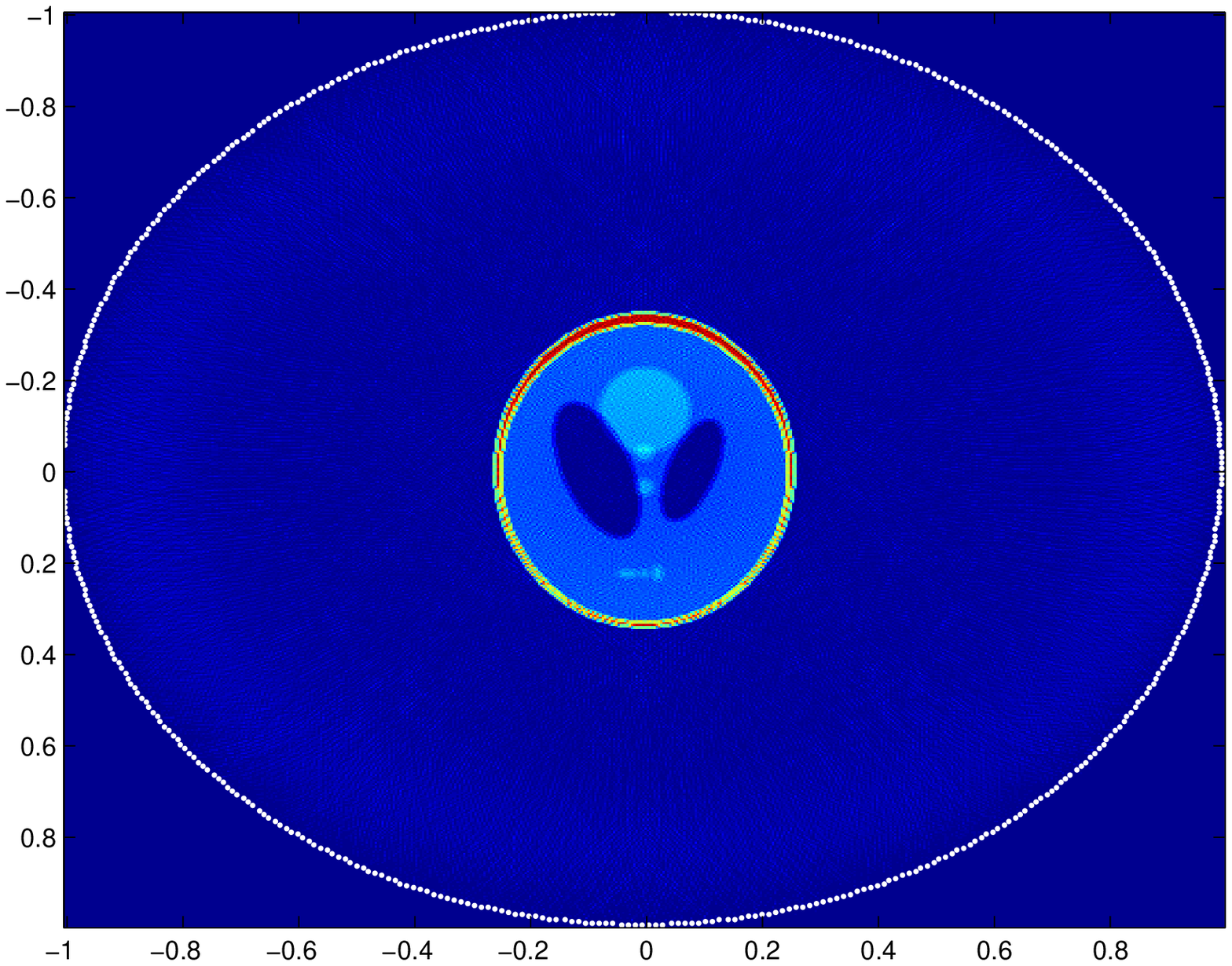}
\qquad
\includegraphics[width = 0.45\textwidth]{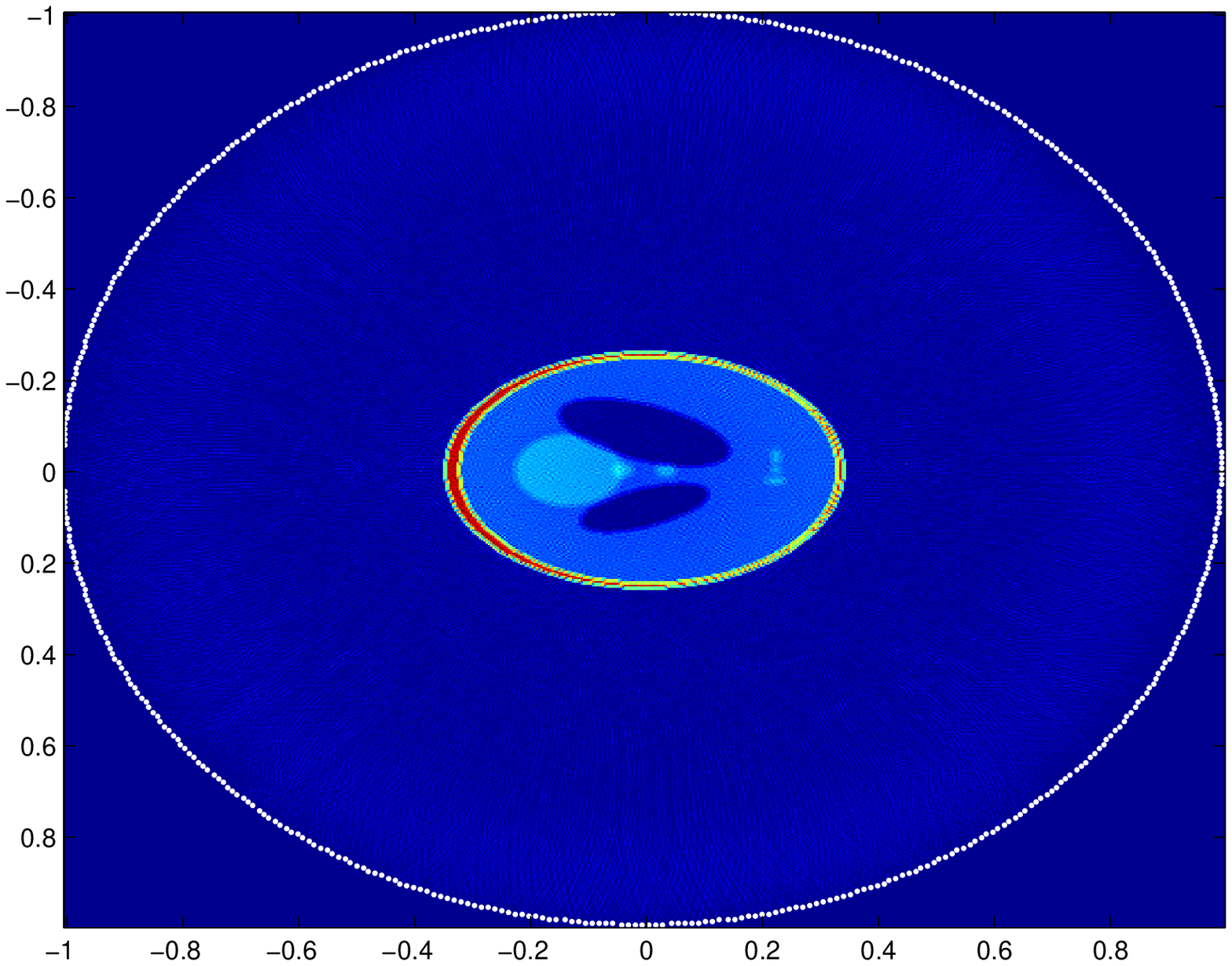}

\caption{Reconstruction of the initial current source $\J=(J_1,J_2,0)$ in a non-attenuation medium using time reversal functional $\Itr$. Top: Initial source, Bottom: Reconstruction. Left: First component of source density, Right: Second component of source density.}
\label{f1}
\end{center}
\end{figure}

\begin{proof}[\textbf{Example 2}]

We choose $\Omega'$ to be a unit disk centered at origin such that $\Omega=\Omega'\times\{0\}$ and  place $512$ equi-distributed sensors on its boundary. We computed  the electric fields over $(\bx,t) \in X\times\big[0,T\big]$ with $l = 4$ and $T=2$ and the space and time discretization steps are taken respectively $\tau = 2^{-13}T$ and $h= 2^{-8}l$. Let the Debye's loss parameter $a$ be $2\times 10^{-4}$. The adjoint wave operator approach for time reversal is tested with cut-off frequencies $\rho=15$ and $\rho=35$ in Figure \ref{f2}. The results clearly indicate an improvement in the resolution in successive reconstructions using $\Ier$ as compared to that using $\Itr$. The images without attenuation correction are blurry whereas the edges in images obtained using $\Ier$ are sharper than those obtained using $\Itr$ and the contrast is relatively higher.  
\end{proof}

\begin{figure}[!tbh]
\begin{center}
\includegraphics[height = 0.20\textheight]{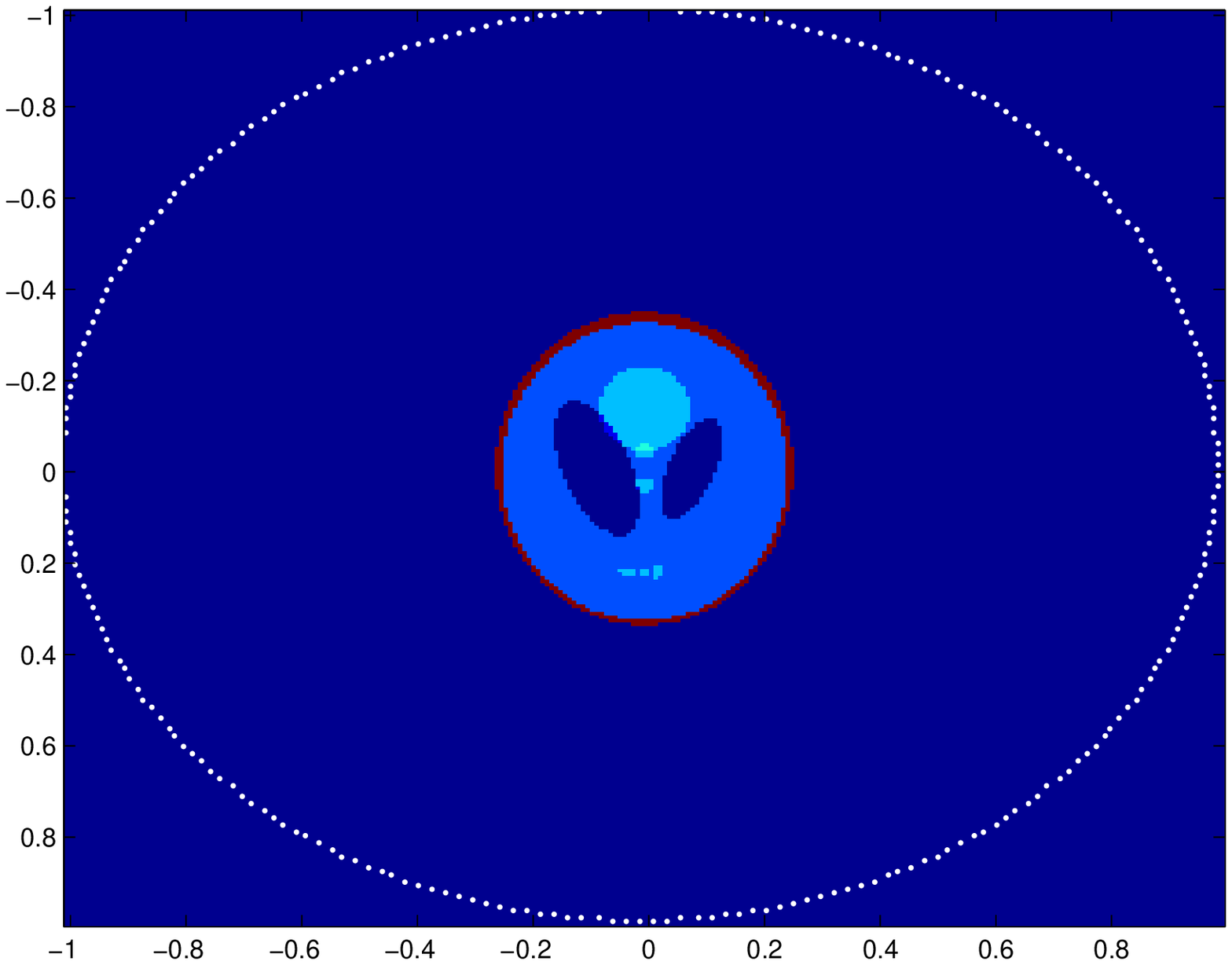}
\qquad
\includegraphics[height = 0.20\textheight]{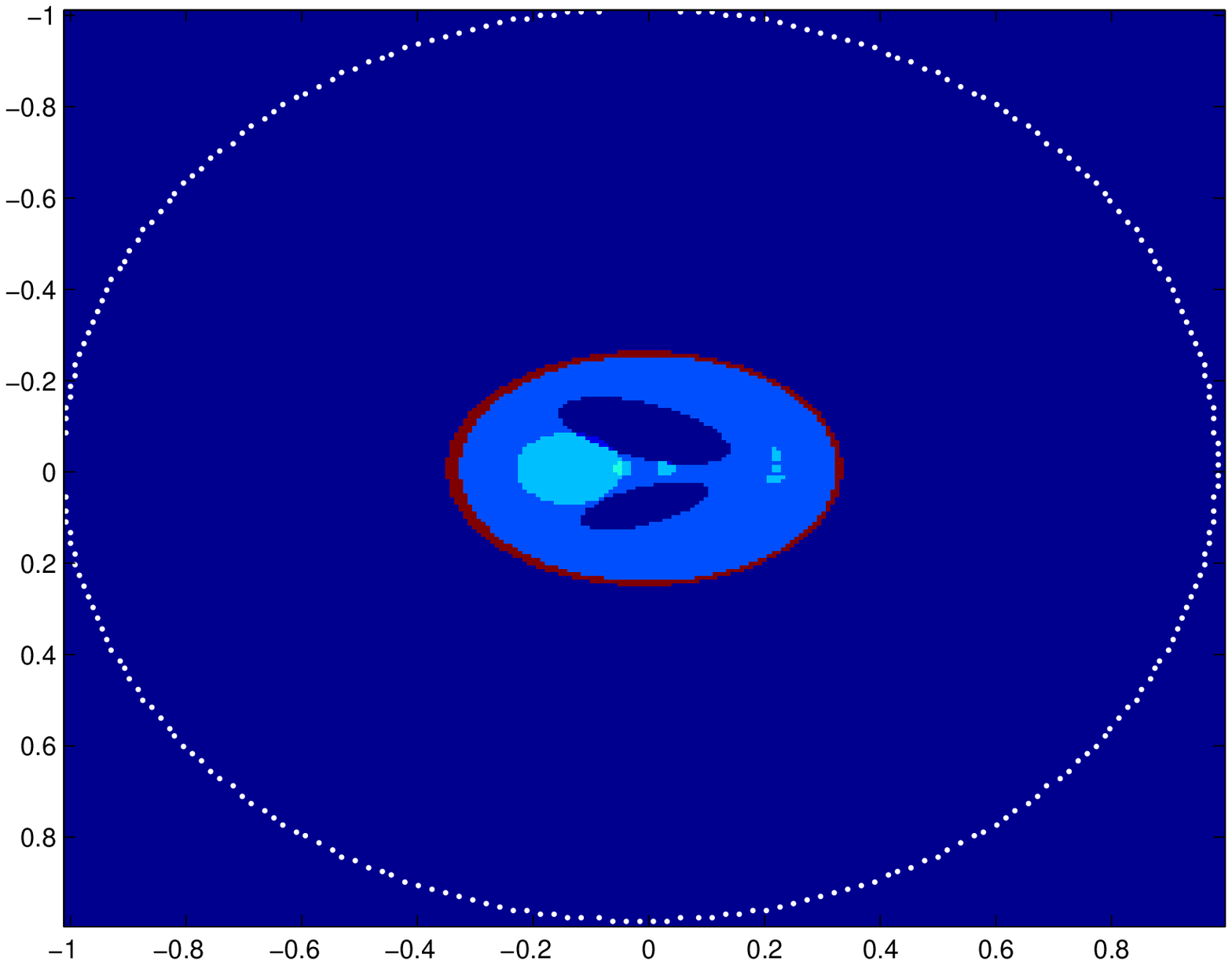}
\\
\includegraphics[height = 0.20\textheight]{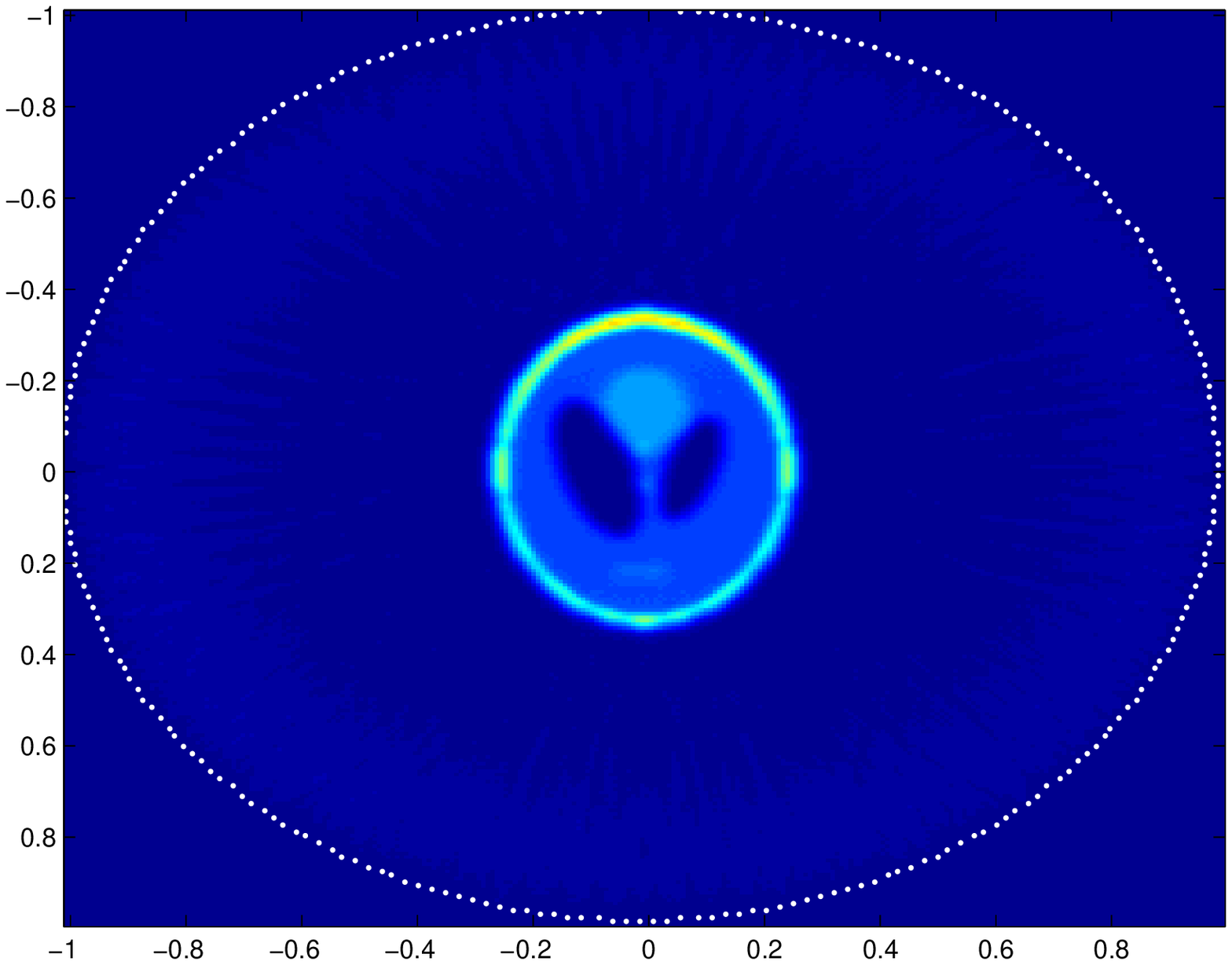}
\qquad
\includegraphics[height = 0.20\textheight]{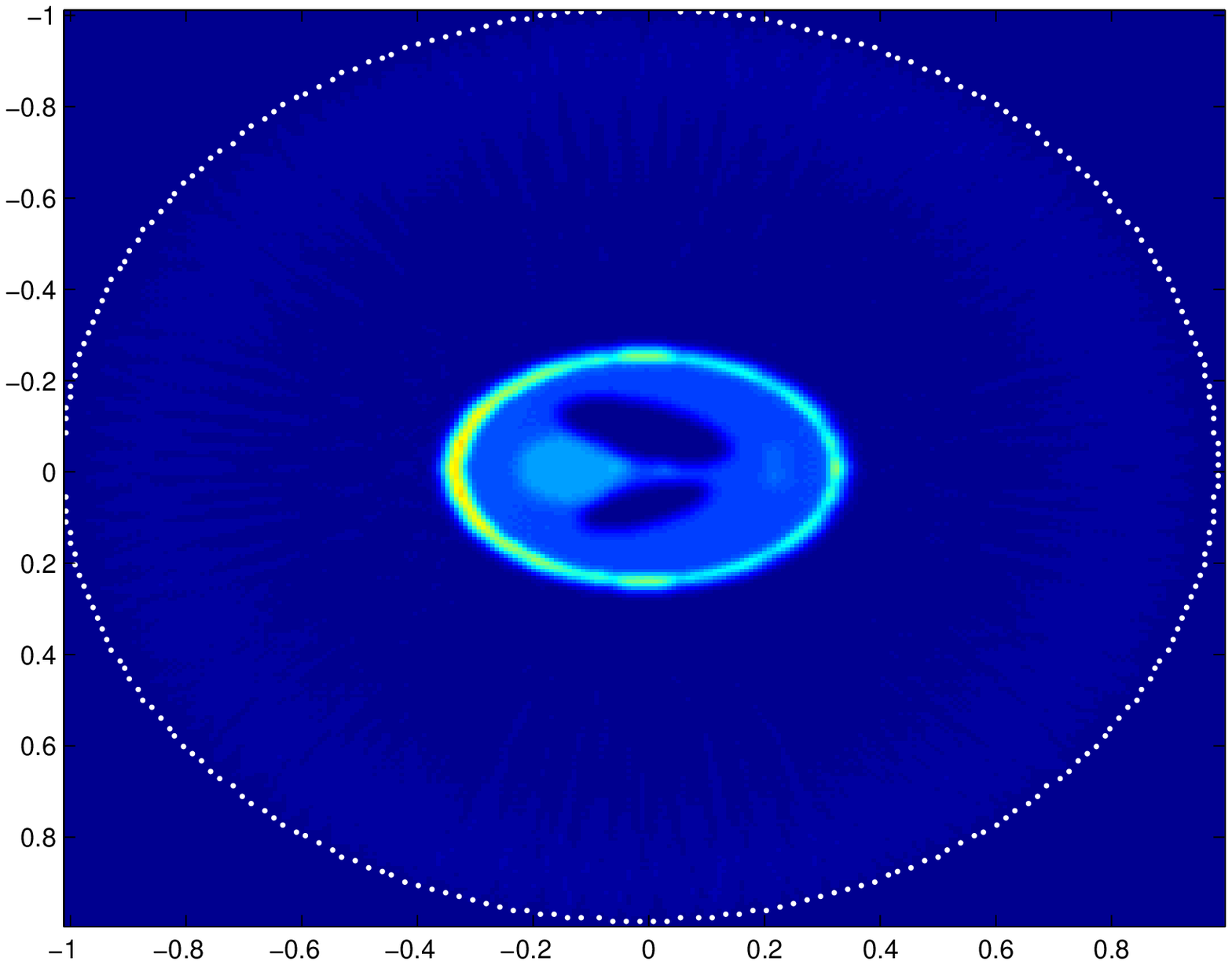}
\\
\includegraphics[height = 0.20\textheight]{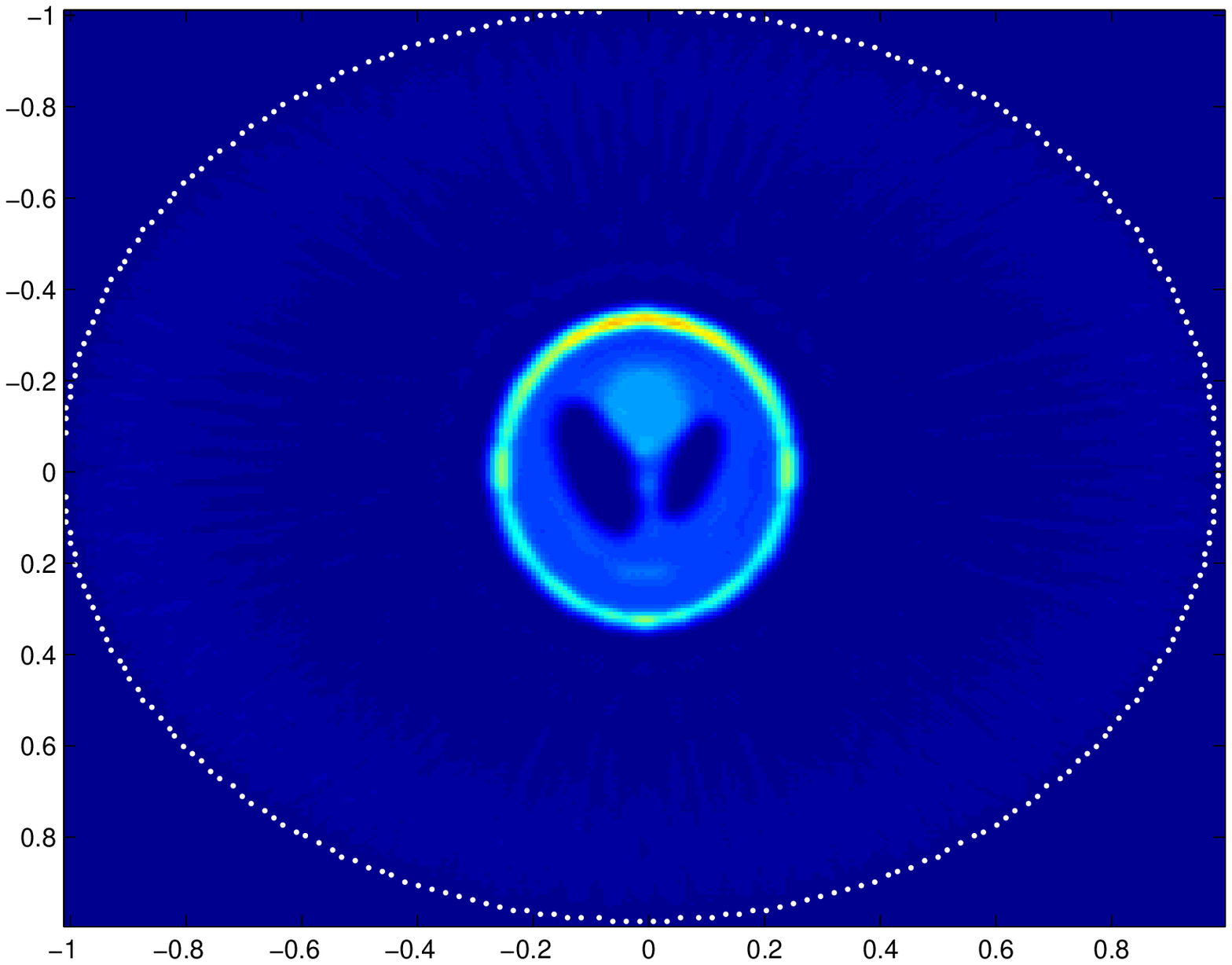}
\qquad
\includegraphics[height = 0.20\textheight]{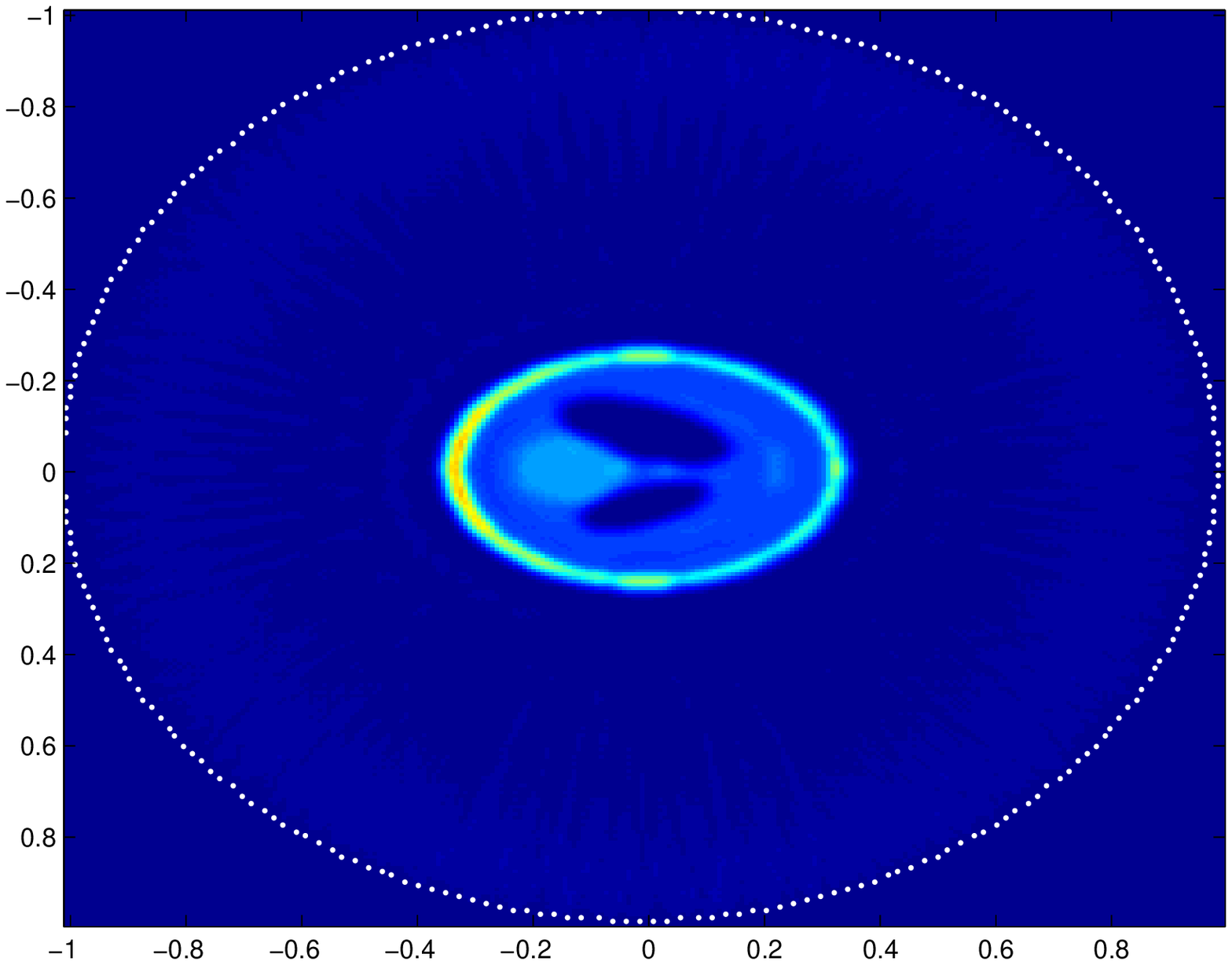}
\\
\includegraphics[height = 0.20\textheight]{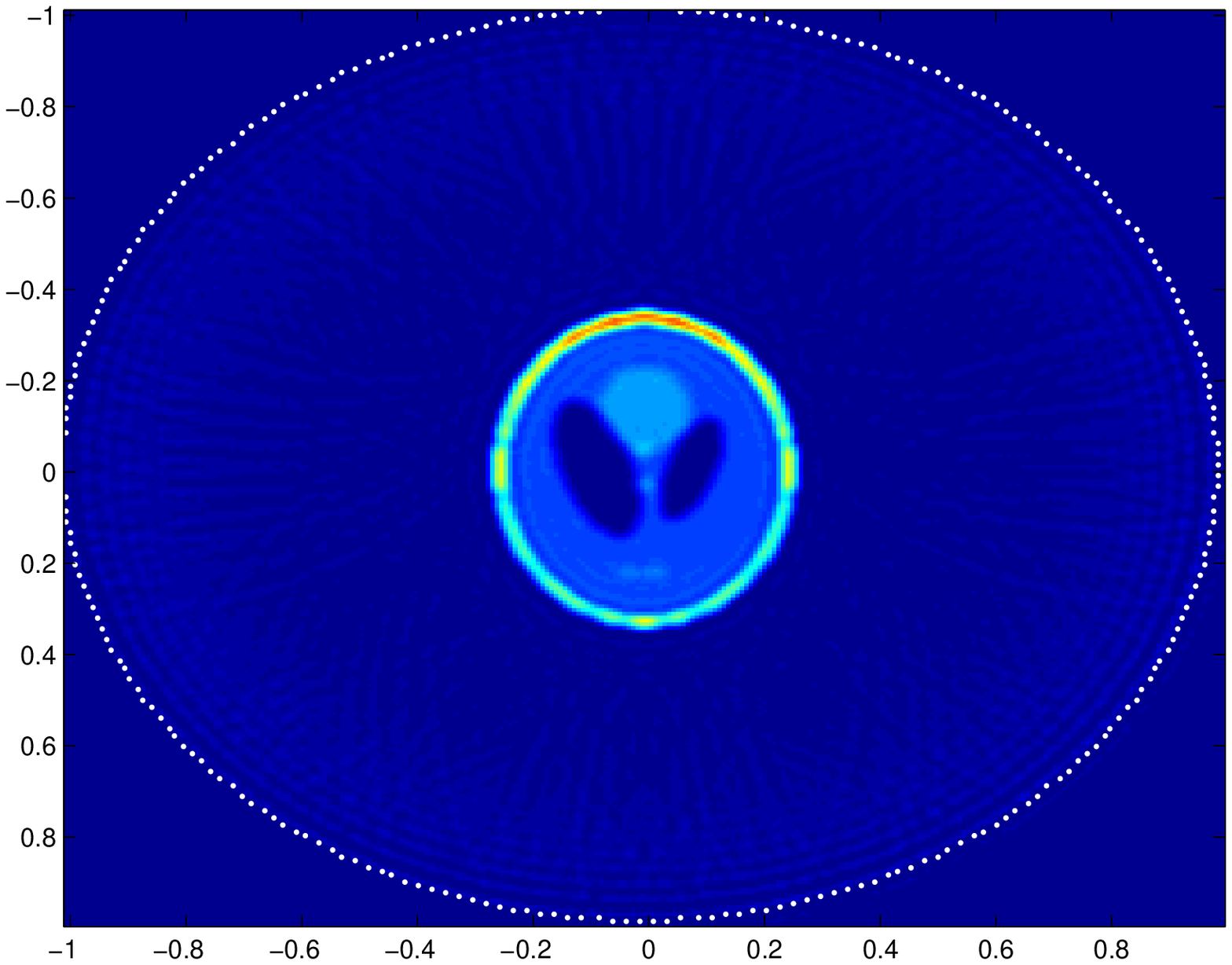}
\qquad
\includegraphics[height = 0.20\textheight]{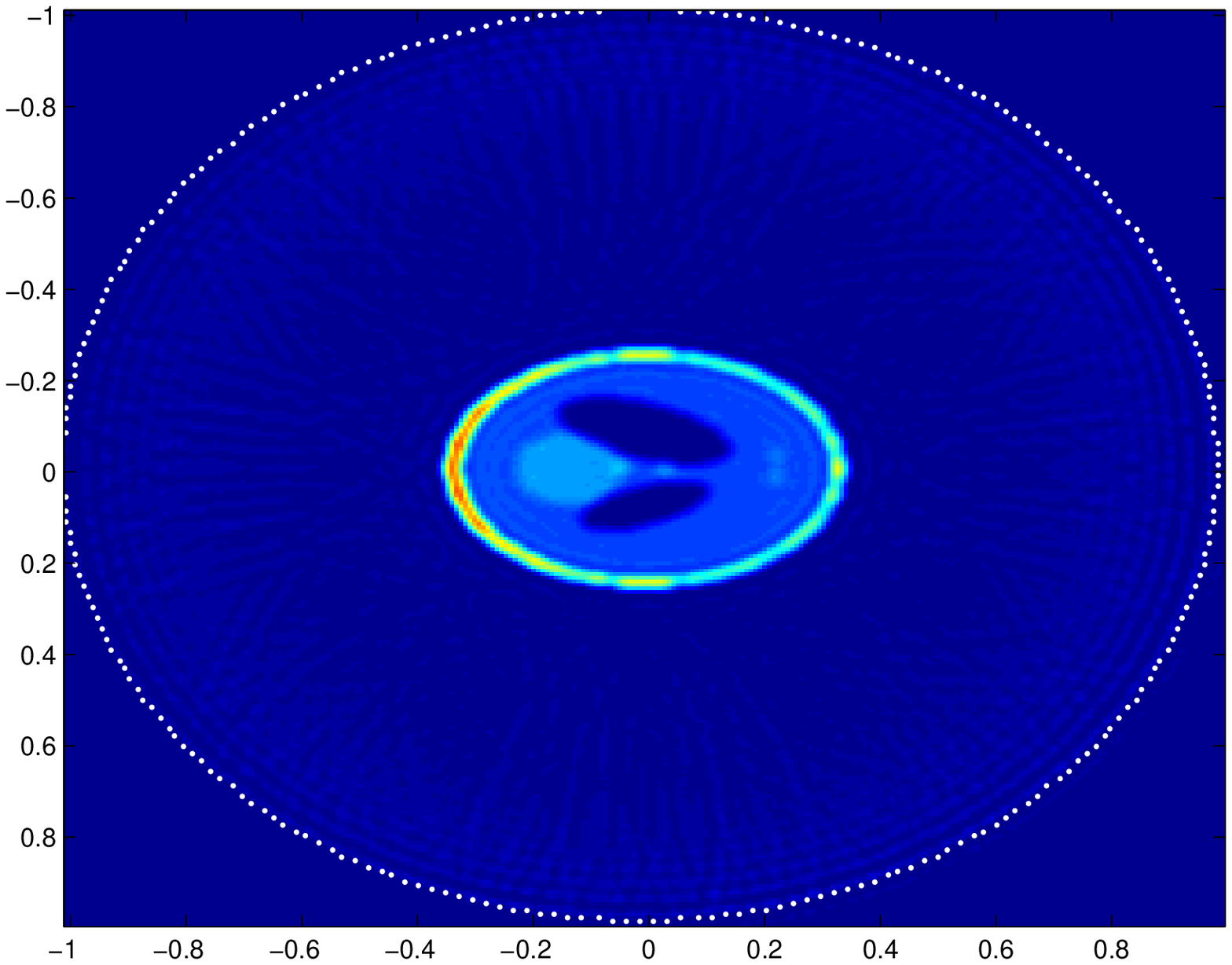}
\\ ~\\
\caption{Reconstruction of the spatial support of electric current source  $\J=(J_1,J_2,0)$ in an attenuating medium  with $a=2\times 10^{-4}$ using time reversal. Left: $J_1$, Right: $J_2$. Top to bottom: Initial source, reconstruction using $\Itr$ (without attenuation correction), reconstruction using $\Ier$ with $\rho=15$ and $\rho= 35$, respectively.}
\label{f2}
\end{center}
\end{figure}

\begin{proof}[\textbf{Example 3}]

Let $\Omega$, $l$, $T$, $\tau$ and $h$ be identical with Example $2$ and the Debye's loss parameter $a$ be $4\times 10^{-4}$. In Figure \ref{f3}, the adjoint wave operator approach for time reversal is tested with cut-off frequencies $\rho=15$ and $\rho=25$. Again, the improvement in the contrast and resolution can be remarked. Albeit, as predicted in the previous sections, increasing the cut-off frequency  induces numerical instability. To this end, the choice of truncation frequency $\rho$, of course as a function of attenuation parameter, is very critical. In this regard, we refer to \cite[Remark 2.3.6]{PhD} for a detailed discussion on the issue and for a threshold value of $\rho$ rendering stability while keeping the resolution intact.
\end{proof}

\begin{figure}[!tbh]
\begin{center}
\includegraphics[height = 0.20\textheight]{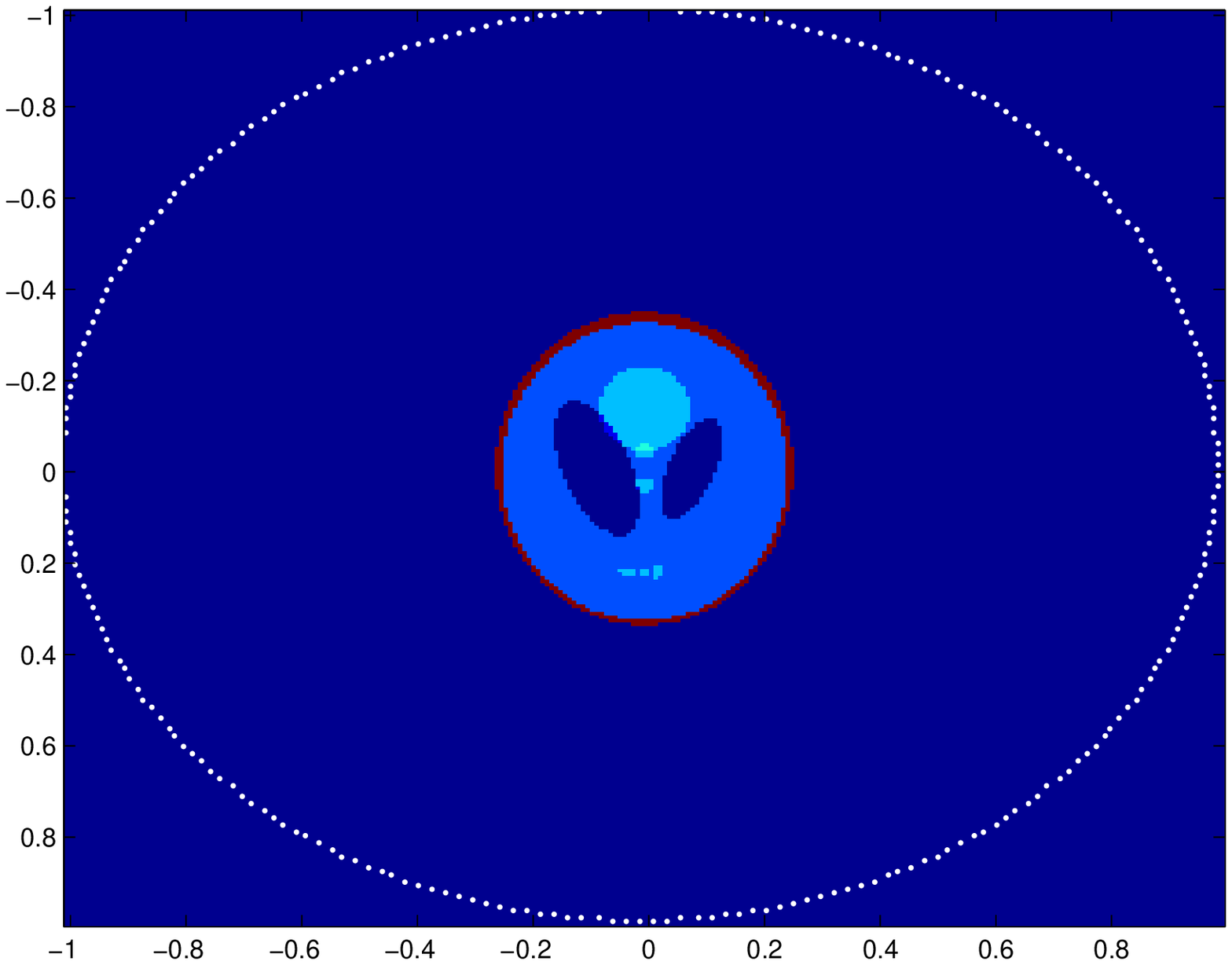} 
\qquad
\includegraphics[height = 0.20\textheight]{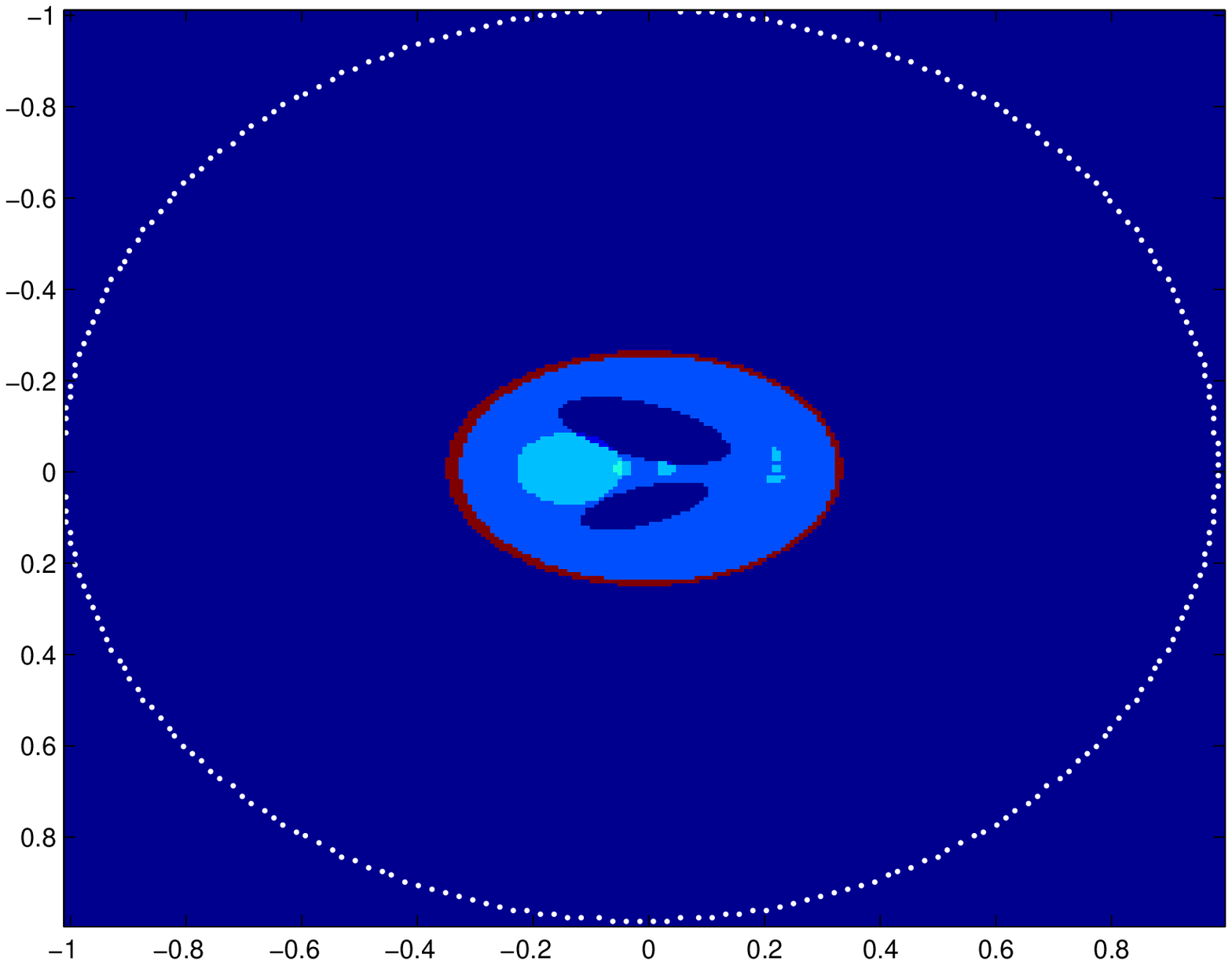} 
\\
\includegraphics[height = 0.20\textheight]{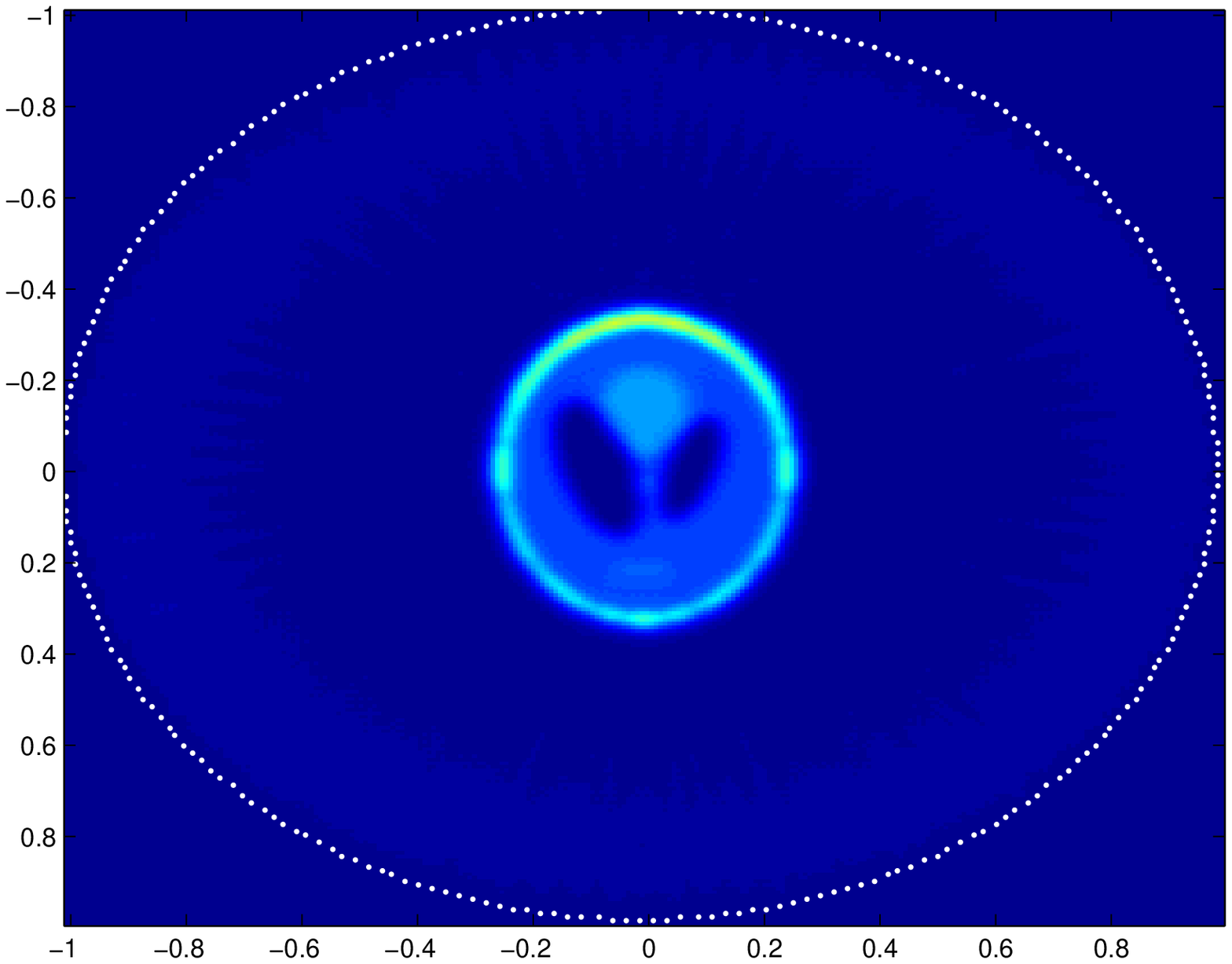}
\qquad
\includegraphics[height = 0.20\textheight]{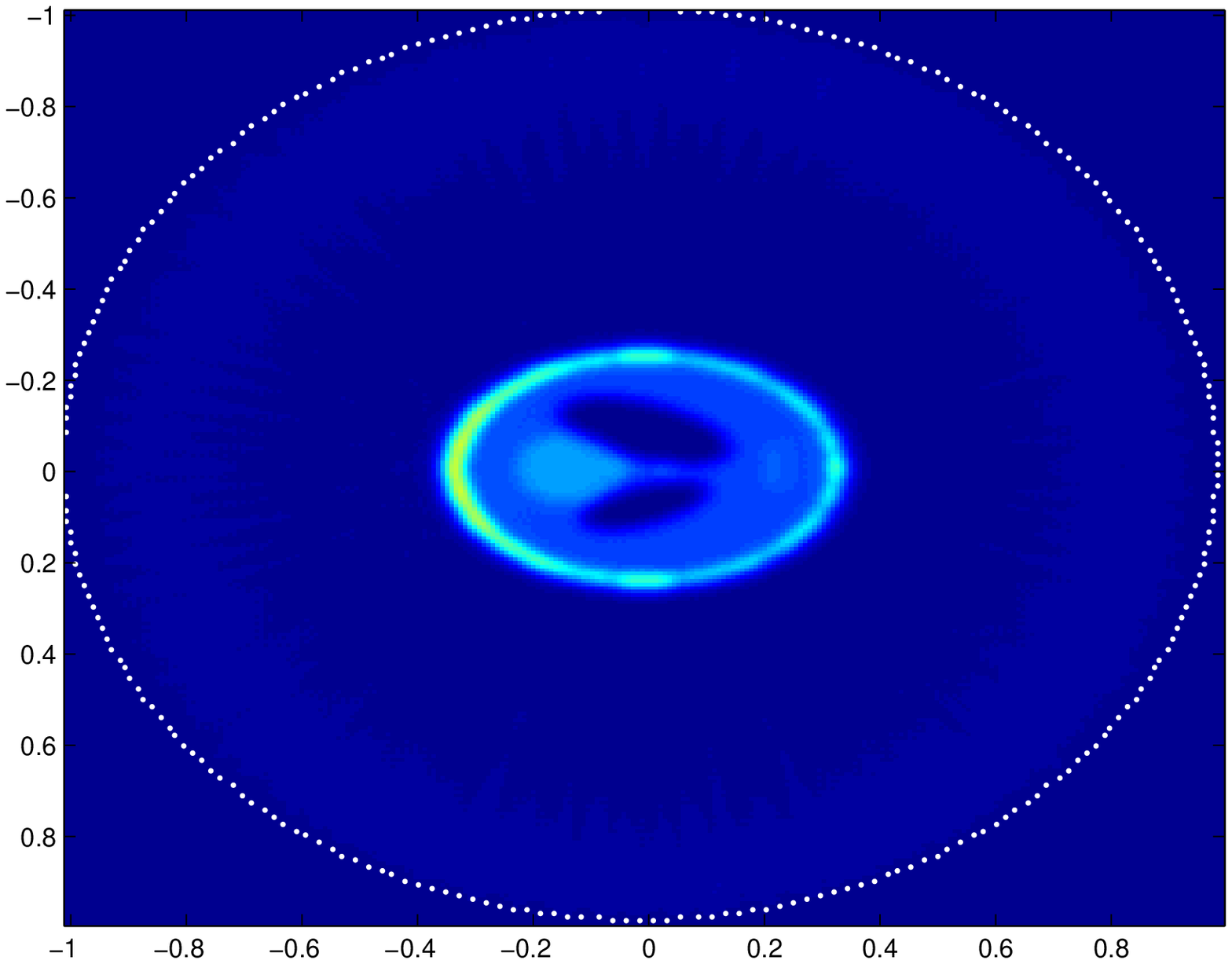} 
\\
\includegraphics[height = 0.20\textheight]{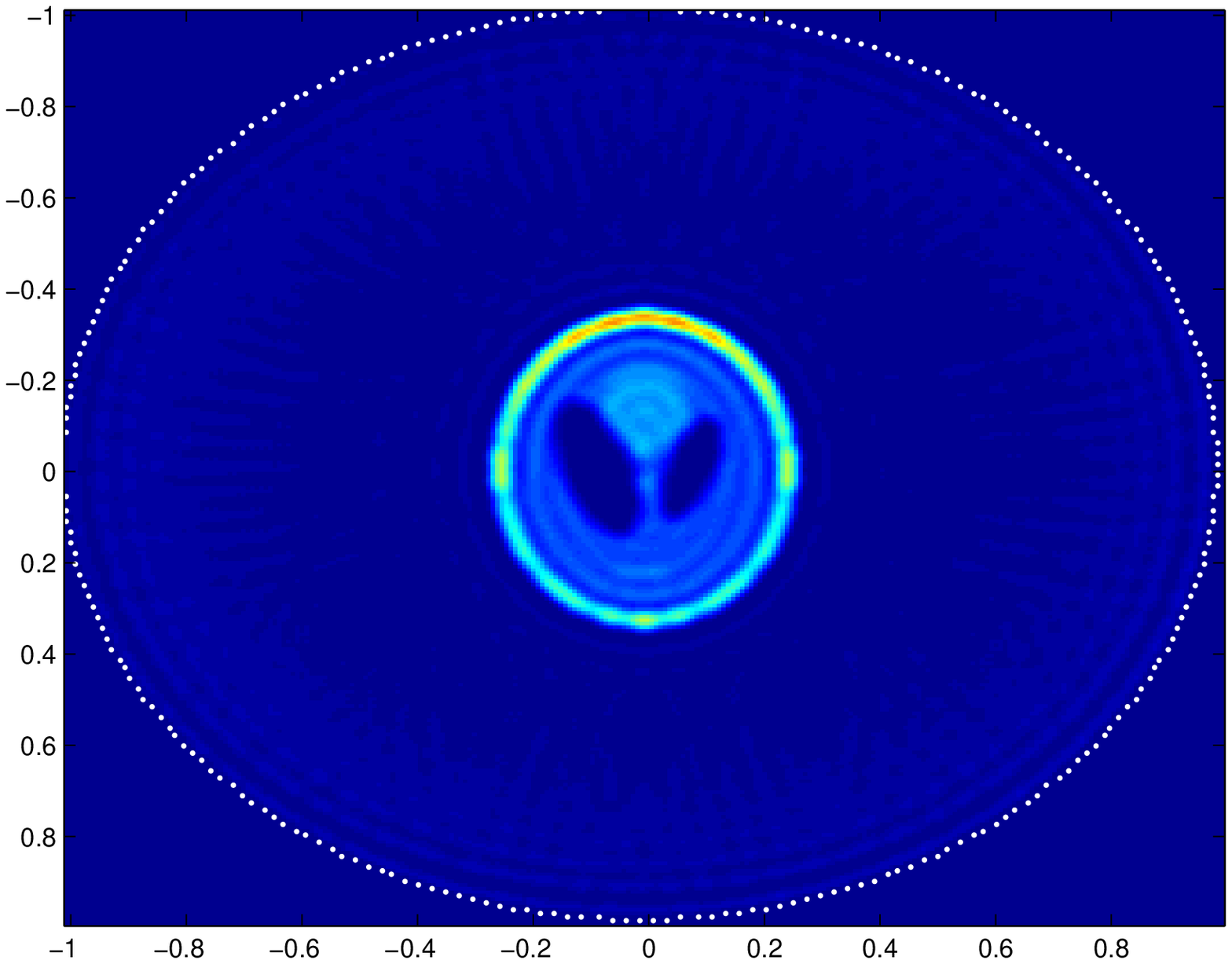}
\qquad
\includegraphics[height = 0.20\textheight]{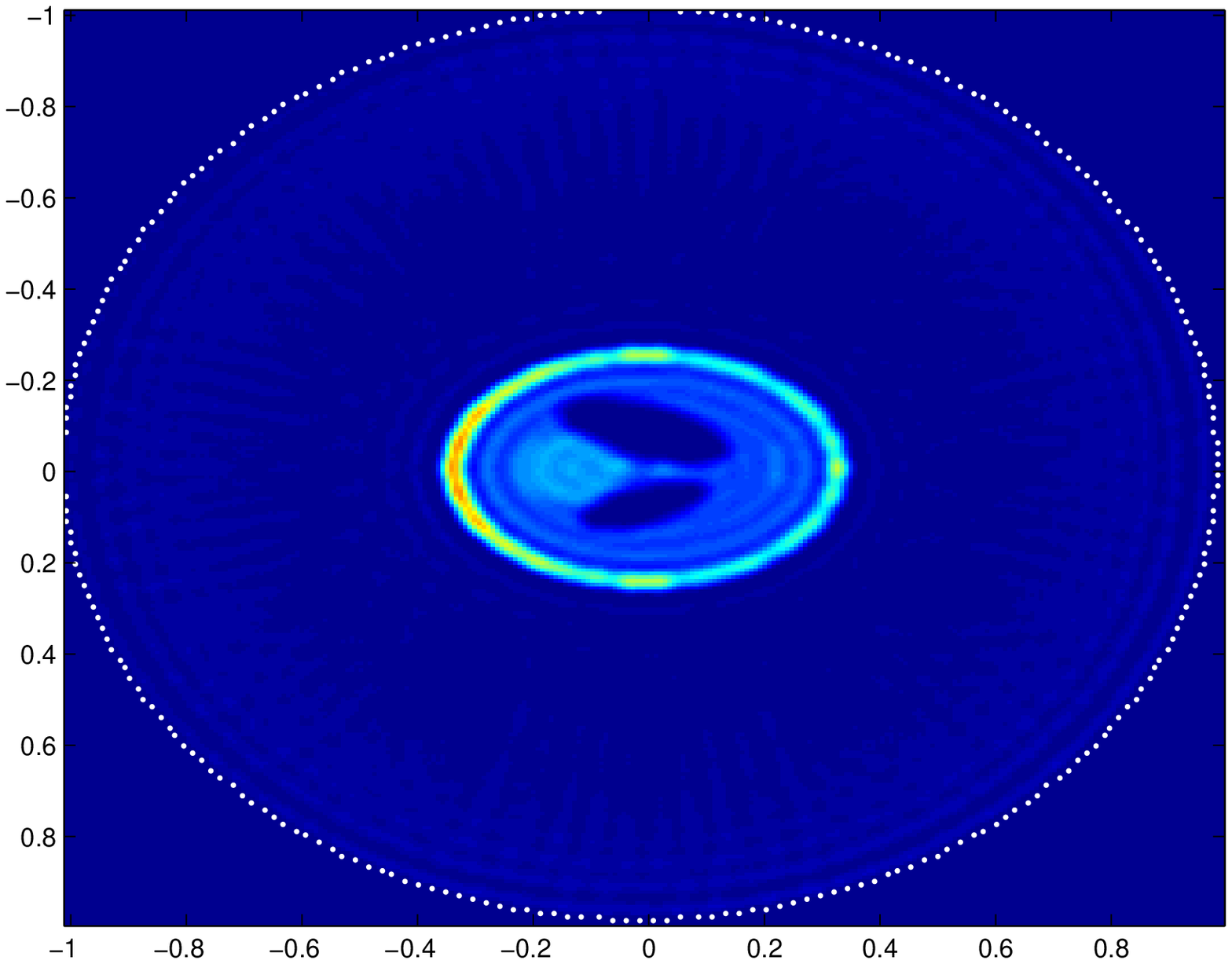} 
\\
\includegraphics[height = 0.20\textheight]{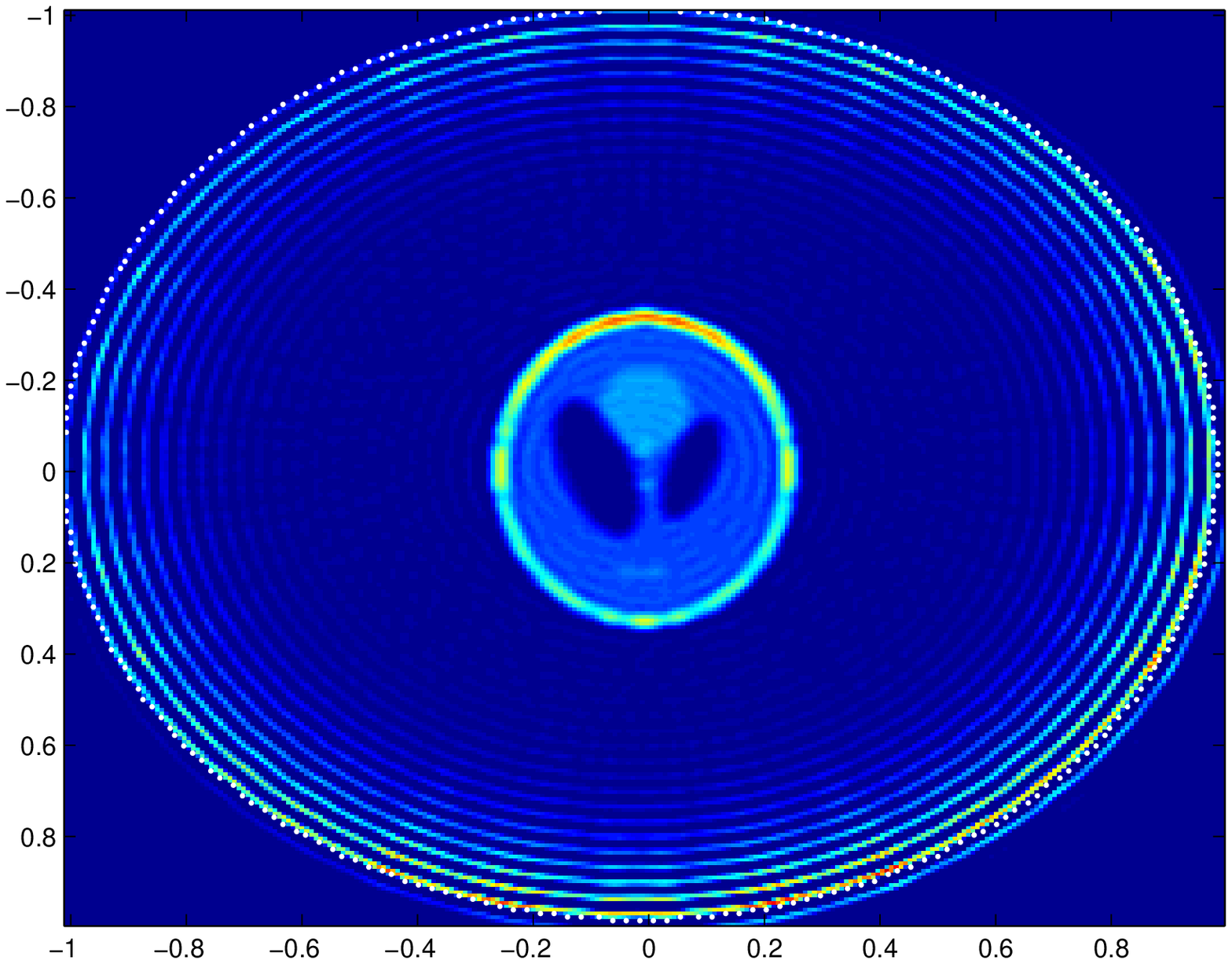}
\qquad
\includegraphics[height = 0.20\textheight]{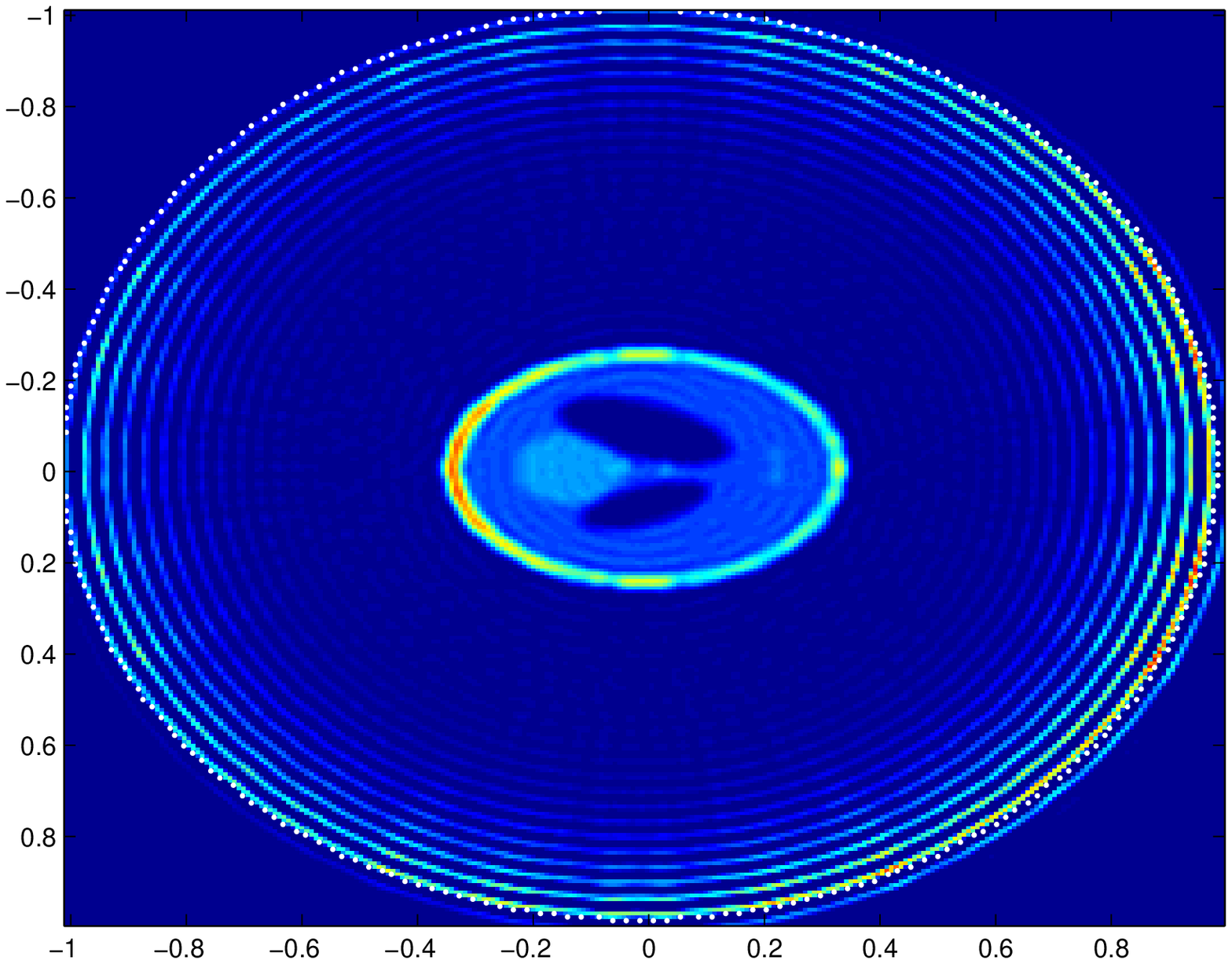}
\\~\\
\caption{Reconstruction of the spatial support of electric current source density $\J=(J_1,J_2,0)$ in an attenuating medium  with $a=4\times 10^{-4}$ using time reversal. Left: $J_1$, Right: $J_2$. Top to bottom: Initial source, reconstruction using $\Itr$ (without attenuation correction), reconstruction using $\Ier$ with $\rho=15$ and $\rho=25$, respectively.}
\label{f3}
\end{center}
\end{figure}

\section{Conclusion}\label{s:conc}

In this investigation, an electromagnetic inverse source problem is tackled using transient boundary measurements of the electric field. A time reversal algorithm  is established for an extended source localization in non-attenuating media. Motivated by this, two more algorithms based on time reversal framework are proposed in order to deal with the problem associated to lossy media wherein the time reversal invariance breaks down. First an adjoint wave back-propagation technique is proposed and justified using asymptotic expansions versus Debye's attenuation parameter of some ill-conditioned attenuation maps. It is proved that this approach yields the current source density up to leading order of attenuation parameter. Then, a second approach is outlined where the lossy data are pre-processed to yield the ideal (non-attenuating) measurements and subsequently the classical time reversal algorithm is invoked to retrieve the current source density. The numerical illustrations clearly indicate the pertinence of the proposed imaging functions. For brevity, the medium is assumed to be non-conducting, but, the results extend to the case otherwise. Time reversal algorithms for inverse scattering problems in lossy dielectric media can be developed similarly and will be the subject of a fourth coming work. 

\bibliographystyle{plain}

\begin{thebibliography}{99}

\bibitem[Fink(1997)]{fink} M. Fink, Time reversed acoustics, {\sl Physics Today}, {\bf 50(3)} (1997), 34.


\bibitem[Borcea et al.(2002)]{borceaPapa02} L. Borcea, G. C. Papanicolaou, C. Tsogka and J. G. Berrymann, Imaging and time reversal in random media, {\sl Inverse Problems}, {\bf 18} (2002),   1247--1279.


\bibitem[Fouque et al.(2007)]{josselin} J. P. Fouque, J. Garnier, G. Papanicolaou and K. S{\o}lna, {\sl Wave Propagation and Time Reversal in Randomly Layered Media}, Springer, (2007).


\bibitem[Lerosey et al.(2005)]{Tele} G. Lerosey, J. de Rosny, A. Tourin, A. Derode, G. Montaldo and M. Fink, Time reversal of electromagnetic waves and telecommunication, {\sl Radio Sci.} {\bf 40} (2005), RS612.


\bibitem[Yavuz and Teixera(2006)]{MY-FT-06} M. E. Yavuz and F. L. Teixera, Full time-domain DORT for ultrawideband electromagnetic fields in dispersive, random inhomogeneous media, {\sl IEEE Transactions on Antennas and Propagation}, {\bf 54(8)} (2006),   2305--2315.

\bibitem[Yavuz and Teixera(2008)]{MY-FT-08} M. E. Yavuz and F. L. Teixera,  On the sensitivity of time-reversal imaging techniques to model perturbations, {\sl IEEE Transactions on Antennas and Propagation}, {\bf 56(3)} (2008), 834--842.

\bibitem[Tanter and Fink(2009)]{FinkTan09TrBioIm}
M.~Tanter and M.~Fink, {Time reversing waves for biomedical applications}, In  {\sl Mathematical Modeling in Biomedical Imaging I}, Lecture Notes in Mathematics, vol. 1983, Springer-Verlag, Berlin, (2009), 73--97.

\bibitem[Xu and Wang(2004)]{XuWang04}
Y.~Xu and L.~V. Wang, {Time reversal and its application to tomography with diffraction  sources},  {\sl Physical Review Letters}, {\bf 92(3)} (2004), Paper ID. 033902.

\bibitem[Gdoura et al.(2012)]{GWL} S. Gdoura, A. Wahab and D. Lesselier, Electromagnetic time reversal and scattering by a small dielectric inclusion, {\sl Journal of Physics: Conference Series},  {\bf 386} (2012), Paper ID. 012010. 
 
\bibitem[Nawaz et al.(2014)]{GG} R. Nawaz, A. Wahab and A. Rasheed,  An intermediate range solution to a diffraction problem with impedance conditions, {\sl Journal of Modern Optics}, {\bf 61(16)} (2007), 1324--1332,

\bibitem[Afzal et al.(2014)]{Afzal} M. Afzal, R. Nawaz, M. Ayub and A. Wahab, Acoustic scattering in flexible waveguide involving step
discontinuity, {\sl PLoS One}, {\bf 9(8)} (2014), Paper ID. e103807.

\bibitem[Carminati(2007)]{Carminati} R. Carminati, R. Pierrat, J. de Rosny and M. Fink, Theory of the time reversal cavity for electromagnetic fields,  {\sl Optics Letters}, {\bf 32(21)} (2007),  3107--3109.

\bibitem[Cassereau and Fink(1992)]{tr2} D. Cassereau and M. Fink, Time-reversal of ultrasonic fields. III. Theory of the closed time-reversal cavity, {\sl IEEE Trans. Ultrasonics, Ferroelect. Freq. Control}, {\bf 39(5)} (1992), 579--592.

\bibitem[Wapenaar(2007)]{Wapenaar} K. Wapenaar,
General representations for wavefields modeling and inversion in geophysics. {\sl Geophysics}, {\bf 75(5)} (2007), SM5--SM17. 


\bibitem[Wahab et al.(0000)]{note}P A. Wahab, A. Rasheed, R. Nawaz and S. Anjum, { Localization of extended current source with finite frequencies}, {\sl Comptes Rendus Math\'ematique}, Doi:10.1016/j.crma.2014.09.009.

\bibitem[Ammari et al.(2013)]{HAetal-11} H. Ammari, E. Bretin, J. Garnier and A. Wahab, 
Time reversal algorithms in viscoeastic media, 
{\sl European Journal of Applied Mathematics}, {\bf 24(4)} (2013),  565-600.

\bibitem[Ammari et al.(2011)]{HAetal-11b} H. Ammari, E. Bretin, J. Garnier and A. Wahab, 
Time reversal in attenuating acoustic media, in {\sl Mathematical and Statistical Methods for Imaging}, Contemporary Mathematics, vol. 548, AMS, (2011),  151--163.

\bibitem[Ammari et al.(2012)]{noise} H. Ammari,  E. Bretin, J. Garnier and A. Wahab, Noise source localization in an attenuating medium,  {\sl SIAM Journal on Applied Mathematics}, {\bf 72(1)} (2012), 317--336.

\bibitem[Wahab and Nawaz(0000)]{noise2} A. Wahab and R. Nawaz, A note on elastic noise source localization, {\sl Journal of Vibration and Control}, doi:10.1177/1077546314546511.

\bibitem[Ammari et al.(2012)]{PAT} H. Ammari, E. Bretin, V. Jugnon and A. Wahab,  Photoacoustic imaging for attenuating acoustic media, in {\sl Mathematical Modeling in Biomedical Imaging II},  Lecture Notes Mathematics, vol. 2035, Springer, (2012), 57-84.

\bibitem[Ammari et al.(0000)]{book1} H. Ammari, E. Bretin, J. Garnier, H. Kang, H. Lee, and A. Wahab, {\sl Mathematical Methods in Elasticity Imaging}, Princeton Series in Applied Mathematics, Princeton University Press, New Jersey, 2014.

\bibitem[Ammari et al.(2014)]{book2} H. Ammari, J. Garnier, W. Jing, H. Kang, M. Lim, K. S{\o}lna, and H. Wang, {\sl Mathematical and Statistical Methods for Multistatic Imaging}, Lecture Notes in Mathematics, vol. 2098, Springer, (2014).

\bibitem[Ammari et al.(2013)]{Localization} H. Ammari, E. Bretin, J. Garnier, W. Jing, H. Kang and A. Wahab,  Localization, stability, and resolution of topological derivative based imaging functionals in elasticity,  {\sl SIAM Journal on Imaging Sciences}, 
{\bf 6(4)} (2013), 2174--2212.

\bibitem[Ammari et al.(2010)]{aggk} H. Ammari, L. Guadarrama Bustos, P. Garapon and H. Kang,  Transient anomaly imaging by the acoustic radiation force,  {\sl Journal of Differential Equations}, {\bf 249} (2010), 1579--1595.

\bibitem[Ammari(2008)]{book} H. Ammari, Introduction to Mathematics of Emerging Biomedical Imaging, {\sl Math. \& App.}, Vol. 62, Springer, (2008).

\bibitem[Valdivia(2012)]{Vald} N. P. Valdivia, Electromagnetic source identification using  multiple frequency information, {\sl Inverse Problems}, {\bf 28} (2012), Article ID 115002.

\bibitem[Michel(2004)]{EEG-rev} C. M. Michel, M. M. Murray, G. Lantz, S. Gonzalez, L. Spinelli and R. Grave de Peralta, EEG source imaging, {\sl Clinical Neurophysiology}, {\bf 115} (2004), 2195--2222.

\bibitem[Porter and Devaney(1982)]{Porter}  R. P. Porter and A. J. Devaney, Holography and the inverse source problem, {\sl Journal of Optical Society of America}. {\bf 72} (1982),  327--330.


\bibitem[Lakhal and Louis(2008)]{source} A. Lakhal and A. K. Louis, Locating radiating sources for Maxwell’s equations using the approximate inverse, {\sl Inverse Problems}, {\bf 24} (2008), Paper ID. 045020.


\bibitem[Bleistein and Cohen(1977)]{Bleistein} N. Bleistein and J. Cohen, Nonuniqueness in the inverse source problem in acoustics and electromagnetics, {\sl Journal of Mathematical Physics}, {\bf 18} (1977), 194--201.

\bibitem[Albanese(2006)]{Albanese} R. Albanese and P. B. Monk, The inverse source problem for Maxwell's equations, {\sl Inverse Problems,} {\bf 22(3)} (2006), Paper ID. 1023.

\bibitem[Bao et al.(2010)]{bao} G. Bao, J. Lin and F. Triki, A multi-frequency inverse source problem, {\sl Journal of Differential Equations}, {\bf 249(12)} (2010),  3443--3465.


\bibitem[Bojarski(1982)]{Bojarski} N. N. Bojarski, A survey of the near-field far-field inverse scattering inverse source integral equation, {\sl IEEE Transactions on Antennas and Propagation}, {\bf 30(5)} (1982), 975--979.


\bibitem[Givoli and Turkel(2012)]{Turkel} D. Givoli and E. Turkel,
Time reversal with partial information for wave refocusing and scatterer identification, {\sl Comput. Methods Appl. Mech. Engrg.}, {\bf 213--216} (2012),   223--242.

\bibitem[Wahab(2011)]{PhD} A. Wahab, {\sl Modeling and Imaging of Attenuation in Biological Media}, PhD Thesis, Centre de Math\'ematiques Appliqu\'ees, \'Ecole Polytechnique, France, (2011).

\bibitem[Kalimeris and Scherzer(2013)]{otmarK} K. Kalimeris and O. Scherzer, Photoacoustic imaging in attenuating acoustic media based on strongly causal models,  {\sl Mathematical Methods in the Applied Sciences}, {\bf 36(16)} (2013), 2254--2264.

\bibitem[Kower(2014)]{Kower} R. Kower, On time reversal in photoacoustic tomography for tissue similar to water, {\sl SIAM Journal on Imaging Sciences}, {\bf 7(1)} (2014),  509--527.

\bibitem[Bretin et al.(2011)]{M2AS} E. Bretin, L. Gaudarrama Bustos and A. Wahab, On the Green function in visco-elastic
media obeying a frequency power law, {\sl Mathematical Methods in the Applied Sciences}, {\bf 34(7)} (2011), 819--830.

\bibitem[H\"ormander(2003)]{hormander} L.~Hormander, {\sl The Analysis of the Linear Partial Differential Operators I: Distribution Theory and Fourier Analysis}, Classics in Mathematics, Springer-Verlag, Berlin, (2003).

\bibitem[Hansen and Yaghjian(1999)]{Hansen} T. B. Hansen and A. D. Yaghjian, {\sl Plane-Wave Theory of Time-Domain Fields: Near-Field Scanning Applications}, IEEE Press, (1999). 


\bibitem[Nedelec(2001)]{nedelec} J.~C. Nedelec,
{\sl Acoustic and Electromagnetic Equations: Integral Representations for Harmonic Problems}, Applied Mathematical Sciences, vol. 144, Springer Verlag, (2001).

\bibitem[Wahab et al.(2014)]{IPSE} A. Wahab, A. Rasheed, R. Nawaz, S. Anjum, Electromagnetic source localization with finite set of frequency measurements, arXiv:1403.5184 [math-ph], 2014.


\bibitem[Koledintseva et al.(2002)]{Koledintseva} M. Y. Koledintseva, K. N. Rozanova, A. Orlandi and J. L. Drewniak, {Extraction of Lorentzian  and Debye parameters fo dielectric and magnetic dispersive materials for FDTD modeling}, {\sl Journal of Electrical Engineering}, {\bf 153(9/S)} (2002), 97--100.

\bibitem[Titchmarsh(1948)]{titchmarsh} E.~C. Titchmarsh,
{\sl Introduction to the Theory of Fourier Integrals}, Clarendon Press, Oxford, (1948).


\bibitem[Canuto et al.(1987)]{spectral} 
C. Canuto, M. Y. Hussaini, A. Quarteroni, and T. A. Zang, 
{\sl Spectral Methods in Fluid Dynamics},
Springer-Verlag, New York-Heidelberg-Berlin, (1987).


\bibitem[Hastings et al.(1996)]{pml} 
F. Hastings, J. B. Schneider, and S. L. Broschat, 
Application of the perfectly matched layer (PML) absorbing boundary condition to elastic wave propagation, 
{\sl Journal of Acoustical Society of America}, {\bf 100} (1996), 3061--3069.

\bibitem[Strang(1968)]{strang}  G. Strang, 
On the construction and comparison of difference schemes, 
{\sl SIAM Journal on Numerical Analysis}, {\bf 5} (1968), 506--517.

\end{thebibliography}

\end{document}